\documentclass[11pt,twoside]{article}

\usepackage{verbatim}
\usepackage{mathrsfs}
\usepackage{amsmath,amssymb,amsthm}
\usepackage{color} 
\usepackage{bbm}
\usepackage{latexsym}
\usepackage{dsfont}
\usepackage{tikz}
\usetikzlibrary{positioning, arrows}
\usepackage{times}

\usepackage{graphicx} 
\usepackage{natbib}
\bibpunct[; ]{(}{)}{,}{a}{}{;}

\usepackage{authblk}

	\graphicspath{ {../Figures/}}

\usepackage{url}
\definecolor{bl}{RGB}{20,20,150}

\newcommand{\RBM}{\operatorname{RBM}}

\newcommand{\Bcal}{\mathcal{B}}

\newcommand{\Ecal}{\mathcal{E}}

\newcommand{\Hcal}{\mathcal{H}}
\newcommand{\Mcal}{\mathcal{M}}

\newcommand{\Scal}{\mathcal{S}}

\newcommand{\Ascr}{\mathscr{A}}
\newcommand{\Cscr}{\mathscr{C}}

\newcommand{\Sscr}{\mathscr{S}}
\newcommand{\Wscr}{\mathscr{W}}

\newcommand{\Mfrak}{\mathfrak{M}}

\newcommand{\R}{\mathbb{R}}

\newcommand{\supp}{\operatorname{supp}}

\newcommand{\rank}{\operatorname{rank}}
\newcommand{\Beh}{\operatorname{Beh}}
\newcommand{\aff}{\operatorname{aff}}

\newtheorem{theorem}{Theorem}
\newtheorem{lemma}{Lemma}

\newtheorem{proposition}{Proposition}

\theoremstyle{definition}
\newtheorem{example}{Example}
\newtheorem{definition}{Definition}

\newtheorem{problem}{Problem}

\renewenvironment{abstract}
{\centerline{\large\bf Abstract}\vspace{0.7ex}%
	\bgroup\leftskip 20pt\rightskip 20pt\small\noindent}%
{\par\egroup\vskip 0.25ex}

\newenvironment{keywords}
{\bgroup\leftskip 20pt\rightskip 20pt \small\noindent{\bf Keywords:} }%
{\par\egroup\vskip 0.25ex}

\title{\vspace{-3mm} 
	A Theory of Cheap Control in Embodied Systems
	}

\author[1]{Guido Mont\'ufar} 
\author[1]{Keyan Ghazi-Zahedi} 
\author[1,2,3]{Nihat Ay} 
\affil[1]{\small Max Planck Institute for Mathematics in the Sciences, Inselstra\ss e 22, 04103 Leipzig, Germany}
\affil[2]{\small Department of Mathematics and Computer Science, Leipzig University, 
	04009 Leipzig, Germany }
\affil[3]{\small Santa Fe Institute, 1399 Hyde Park Road, Santa Fe, NM 87501, USA}
\date{}

\usepackage{fancyhdr}
\begin{document}

\thispagestyle{empty}
\maketitle

\begin{abstract}
We present a framework for designing cheap control architectures for embodied agents. Our derivation is guided by the classical problem of universal approximation, whereby we explore the possibility of exploiting the agent's embodiment for a new and more efficient universal approximation of behaviors generated by sensorimotor control. This embodied universal approximation is compared with the classical non-embodied universal approximation. To exemplify our approach, we present a detailed quantitative case study for policy models defined in terms of conditional restricted Boltzmann machines. In contrast to non-embodied universal approximation, which requires an exponential number of parameters, in the embodied setting we are able to generate all possible behaviors with a drastically smaller model, thus obtaining cheap universal approximation. We test and corroborate the theory experimentally with a six-legged walking machine. The experiments show that the sufficient controller complexity predicted by our theory is tight, which means that the theory has direct practical implications.
\end{abstract}


\section{Introduction}
In artificial intelligence, learning is one of the central fields of interest.
Crucial for the success of any learning method is the complexity of the
underlying model, e.g.~a neural network. If the model is chosen too complex, the
learning algorithm will likely require too much time and get stuck in a
suboptimal solution. If it is chosen too simple, it might not be able to solve
the problem at all. It is known from biological systems, that the exploitation
of the body and environment allows a reduction of the neural system's complexity
\citep{Pfeifer2006aHow-the-Body}.

The goal of this article is to provide a framework that allows to determine the complexity
of a control architecture in accordance with the \emph{cheap
design} principle from embodied artificial intelligence~\citep{Pfeifer2006aHow-the-Body}. 
Cheap design in this context refers to the relatively low complexity of the brain or
controller in comparison with the complexity of an observed behavior.
A classical example is given by the Braitenberg
vehicles~\citep{Braitenberg1984aVehicles}, which are \emph{Gedankenexperiments}
designed to show how a seemingly complex behavior can result from very simple
control structures. Braitenberg discusses several artificial
creatures with simple wirings between sensors and actuators. 
He then describes how these systems produce a behavior that an external observer
would classify as complex if the internal wirings were not revealed. Most interestingly, he then
relates the wiring of his vehicles to various neural structures in the human brain. The idea
of a simple wiring that leads to complex behaviors is also discussed by~\citet{Pfeifer2006aHow-the-Body}, who present the walking behavior of an ant as an example. 
Without taking the embodiment and, in particular, the sensorimotor loop into
account, the complex behavior (of a complex morphology) seems to require a
complex control structure~\citep[p.~79]{Pfeifer2006aHow-the-Body}.
A strong indication that cheap design is a common principle in biological
systems is given by the fact that the human brain accounts for only 2\% of the
body mass but is responsible for 20\% of the entire energy consumption
\citep{Clark1999aCirculation}, which is also remarkably constant
\citep{SOKOLOFF1955aEFFECT}.
Further support for cheap design as a common principle is given by a recent
study on the brain sizes of migrating birds. It is known that migrating birds
have a reduced brain size compared with their resident relatives.
\citet{Sol2010aEvolutionary} have studied various species and the affected brain
regions and point out that the reduced brain sizes could be a direct result from
the need to reduce energetic, metabolic and cognitive costs for migrating
birds. 

One way to achieve cheap design in this context is described as
\emph{compliance} in the embodied artificial intelligence community. A system is
described as compliant, if it not only copes with the hard physical constraints
it is subject to, but if it exploits them in order to minimize the required
control effort. An illustrative example is the human walking behavior, which
only needs to be actively controlled during the stance phase. The swing phase
results mainly from the interaction of the physical properties of the leg with
the environment (gravity). This is demonstrated by the Passive Dynamic
Walker~\citep{McGeer1990aPassive}, which is a purely mechanical system that
resembles the physical properties of human legs. The human walking behavior is
emulated as a result of the interaction of the mechanical system with its
environment (gravity and a slope). It is an impressive example of cheap design
that requires no active control \mbox{at all}.

We are interested in quantifying to what extent a control structure can be
reduced if the physical constraints are taken into account. Above, we referred
to a system as cheaply designed, if it has a control structure of low complexity
produces behaviors which an external observer would classify as complex. In this
work, we are not concerned with the complexity of the behavior. Instead,
we present an approach to determine the minimal complexity of a control
structure that is able to produce a given set of desired behaviors (which can
also be all theoretically possible behaviors) with a given morphology in a given
environment. In other words, rather than comparing the complexities of the
control structure and the behavior, we ask: what is the minimal brain complexity
(or size) that can control all (desired) behaviors that are possible with the
body and environment in which it is embedded? 

There are various different complexity measures available in literature, of
which the predictive information \citep{Bialek1999aPredictive}, relevant
information \citep{Polani2006aRelevant}, and the Kolmogorov complexity
\citep{Schmidhuber2009aDriven}, are just a few examples.  All these approaches
have their specific strengths. However, they do not explicitly quantify how much
the controller complexity can be reduced as a result of the agent’s embodiment,
which is the focus of this work.

We follow a bottom-up, understanding by building
approach~\citep{Brooks1991aIntelligence} to cognitive science, which is also
known as behavior-based robotics~\citep{Brooks1991bIntelligence} and embodied
artificial intelligence~\citep{Pfeifer2006aHow-the-Body}. The core concept is
that cognitive systems are considered as embedded and situated agents which
cannot be understood if they are detached from the sensorimotor loop. This
implicitly means that we assume sensor state sparsity and continuity of physical
constraints. Consider the human retina as an example. We do not see random images
but structured patterns and the sequence of these patterns is also highly
dependent on our behavior. This behavior-dependent structuring of information is
also know as \emph{information self-structuring} and it has been identified as
one of the key principles of learning and
development~\citep{Lungarella2005aInformation,Pfeifer2007aSelf-Organization}.
The second implication from the sensorimotor loop is continuity, e.g.~natural
systems are unable to teleport themselves from one place to another. Therefore,
we can safely assume that the world around us will not be too different from the
recent past and the recent future.

The sensorimotor loop
(SML)~\citep{klyubin2004tracking,Ay2014On-the-causal-structure-of-the-sensorimotor}
is described by a type of partially observable Markov decision process (POMDP)
where an embodied agent chooses actions based on noisy partial observations of
its environment. An illustration of this causal structure is given in
Figure~\ref{figure:SML}. We aim at optimizing the design of policy models for
controlling these processes. One aspect of the optimal design problem is
addressed by working out the optimal complexity of the policy model. In
particular, we are interested in the minimal number of units or parameters
needed in order to obtain a network that can represent or approximate a desired
set of behaviors within a given degree of accuracy. A first step towards
resolving this problem is to address the minimal size of a universal
approximator. In realistic scenarios, universal approximation is out of
question, since it demands an enormous number of parameters -- many more than
actually needed. In this paper we reconsider the universal approximation problem
by exploiting embodiment constraints and restrictions in the desired behavioral
patterns. 

We introduce the notions of {\em embodied behavior dimension} and {\em embodied
universal approximation}, which quantify the effective dimension of a system
that is subject to sensorimotor constraints (embodiment) and formalize the
minimal control paradigm of cheap design in the context of the sensorimotor
loop. We substantiate these ideas with theoretical results on the
representational capabilities of conditional restricted Boltzmann machines
(CRBMs) as policy models for embodied systems. CRBMs are artificial stochastic
neural networks where the input and output units are connected bipartitely and
undirectedly to a set of hidden units. Given the embodied behavior dimension, we
derive bounds on the number of hidden units of CRBMs, that suffices to generate
all possible behaviors by appropriate tuning of interaction weights and biases.
In order to test our theory, we present an experimental study with a six-legged
walking robot, and find a clear corroboration of our theorems. The experiments
show that the sufficient controller complexity predicted by our theory is tight,
which means that the theory has direct practical implications.

CRBMs are defined by clamping an {\em input} subset of the visible units of a 
Restricted Boltzmann machine (RBM)~\citep{Smolensky1986,freund1994unsupervised}. 
Conditional models of this kind have found a wide range of applications, e.g., in classification, collaborative filtering, and motion modeling~\citep[see][]{LarochelleB08,Salakhutdinov:2007:RBM,sutskever_hinton_07,Taylor06modelinghuman}, and have proven useful as policy models in reinforcement learning settings~\citep{Sallans:2004:RLF:1005332.1016794}. 
These networks can be trained efficiently~\citep{Hinton:2002:TPE:639729.639730,Hinton:practical} and are well known in the context of learning representations and deep learning~\citep[see][]{BengioLearning}. 
Although estimating the probability distributions represented by RBMs is hard~\citep{LongServedio10}, approximate samples can be generated easily from a finite Gibbs sampling procedure. The theory and in particular the expressive power of RBM probability models has been studied in numerous papers~\citep[e.g.,][]{LeRoux:2008:RPR:1374176.1374187,Montufar2011, NIPS2011_0307,NIPS2013_5020}. 
Recently the representational power of CRBMs has been studied in detail~\citep{montufar2014CRBMs}. 
CRBMs can model non-trivial conditional distributions on high-dimensional input-output spaces using relatively few parameters, and their complexity can be adjusted by simply increasing or decreasing the number of hidden units. 
Hence we chose this model class for illustrating our discussion about the complexity of SML control problems.

This paper is organized as follows. 
Section~\ref{section:preliminaries} contains definitions around the SML. 
Section~\ref{section:embodied} presents the notions of embodied behavior dimension and embodied universal approximation, 
which we use to quantify and enforce dimensionality reduction. 
Section~\ref{section:theory} contains our theoretical discussion on the representational power of CRBM models, 
comparing the non-embodied and the embodied settings, and pointing the role of the embodied behavior dimension. 
Section~\ref{section:experiments} puts the theory to the test in a robot control problem.  
Section~\ref{section:discussion} offers our conclusions and outlook. 
The Appendix contains technical proofs and details about possible generalizations of the discussion presented in the main part of the paper. 

\section{The Causal Structure of the Sensorimotor Loop} \label{section:preliminaries}

What is an embodied agent? In order to develop a theory of embodied agents that allows us to cast the core principles of the field of embodied intelligence into rigorous theoretical and quantitative statements, we need an appropriate formal model. Such a model should be  general enough to be applicable to all kinds of embodied agents, including natural as well as artificial ones, and specific enough to capture the essential aspects of embodiment. How should such a model look like? First of all, obviously, an embodied agent has a body.  This body is situated in an environment with which the agent can interact, thereby generating some behavior. In order to be useful, 
this behavior has to be guided or controlled by the agent's brain or controller. Drawing the boundary between the brain on one side and the body, together with the environment, on the other side suggests a black box perspective of the brain. The brain receives sensor signals from and sends effector or actuator signals to the outside world. All it knows from the world is based on this closed loop of signal transmission. In other words, the world is a black box for the brain with which it interacts through sensing and acting. 
\begin{figure}[h]
  \begin{center}
	\includegraphics[clip=true, trim=0cm .3cm 0cm 0cm, scale=.5]{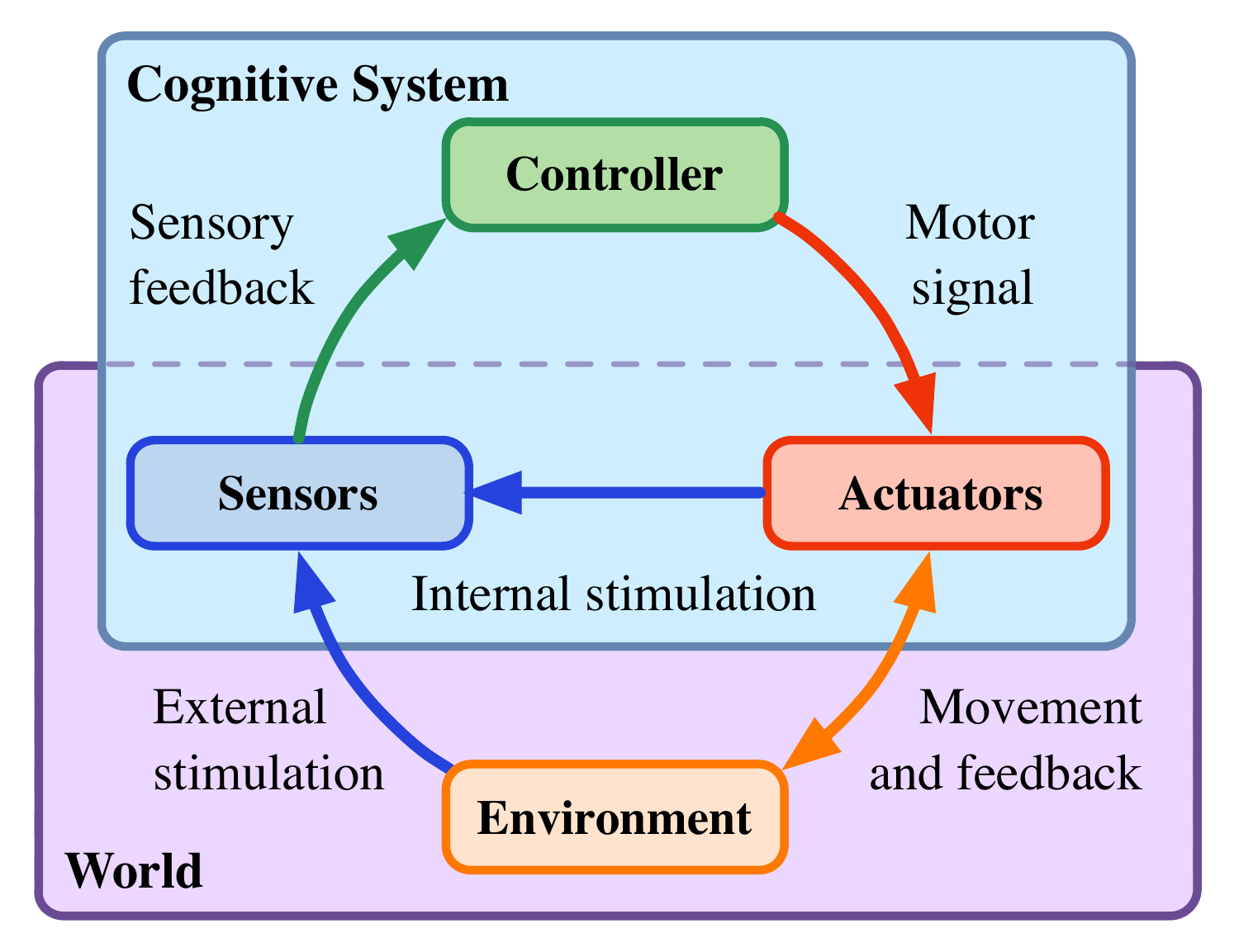}
  \end{center}
  \caption{Sensorimotor loop}
  \label{fig:sml_concept}
\end{figure}
In particular, the boundary between the body and the environment is not directly ``visible'' for the brain. Both are parts of that black box and interact with the brain in an entangled way. Therefore, we consider them as being one entity, the outside world or simply the 
world. The brain is causally independent of the world, given the sensor signals, and the world is causally  independent of the brain, given the actuator signals. This is the black box perspective.  

Let us now develop a formal description of this sensorimotor loop. We denote the set of world states by $\mathscr{W}$. This set can be, for instance, the position of a robot in a static 3D environment. 
Information from the world is transmitted to the brain through sensors. Denoting the set of sensor 
states by $\mathscr{S}$, we can consider the sensor to be an information transmission channel  
from $\mathscr{W}$ to $\mathscr{S}$ as it 
is defined within information theory. Given a world state $w \in \mathscr{W}$, the response of the sensor can be characterized by a probability distribution of possible sensor states $s \in \mathscr{S}$ as result of $w$. For instance, if the sensor is noisy, then its response will not be uniquely determined. If the sensor is noiseless, that is, deterministic, then there will be only one sensor state as possible response to the world state $w$. In any case, the response of the sensor given 
$w$ can be described in the following way: for a set $\mathsf{S}$ of sensor
states we simply say how likely it is that the sensor will respond with a sensor
state $s$ that is contained in $\mathsf{S}$. Formally, we can express this
likelihood by a number $\beta(w; \mathsf{S})$ between zero and one, which is the
probability of $\mathsf{S}$ given $w$. Collecting these numbers for all world states $w
\in \mathscr{W}$ and sets $\mathsf{S}$ leads to the mathematical definition of a channel,
also called a Markov kernel. It can be summarized as a map 
\[
    \beta: \; \Wscr \; \longrightarrow \;  \Delta_{\Sscr} \, ,
\]
where $\Delta_{\Sscr}$ denotes the set of probability distributions on the set $\Sscr$ of sensor states. 
The set of all such sensor channels is denoted 
by $\Delta^{\Wscr}_{\Sscr}$. Given a sensor channel $\beta$, 
there is another way to represent the probability distribution that is assigned
to a world state $w$. Instead of providing a list $\beta(w; \mathsf{S})$ for all
sets $\mathsf{S}$
of interest, we can 
restrict attention to infinitesimally small sets $ds$, leading to the notation
$\beta(w; ds)$. In order to represent $\beta(w;\mathsf{S})$
in this notation, we have to integrate over all the infinitesimal $ds$ in a set
$\mathsf{S}$, that is
\[
  \beta(w; \mathsf{S}) \; = \; \int_\mathsf{S} \beta(w ; ds) \, .
\] 
Whenever the base set $\Sscr$ is discrete, we simply replace $ds$ by $s$
and use $\beta(w; s)$ instead of $\beta(s; ds)$.
Note that, as a Markov kernel, 
$\beta$ has to satisfy various conditions. In order to provide a 
mathematically rigorous treatment, we assume that these conditions are
satisfied. However, in order to improve the readability of the paper, we will
not be very explicit with this~\citep[for the technical definitions see,
e.g.,][]{bauer1996probability}. 

After having described in detail the mathematical model of a sensor, it is now
straightforward to consider corresponding formalization of the other components
of the sensorimotor loop. We continue with the notion of a policy. The agent can
generate an effect in the world in terms of its actuators. Since we consider the
body as part of the world, this can lead, for instance, to some body movement of
the agent. In order guide this movement, it is beneficial for the agent to
choose its actuator state based on the information about the world received
through its sensors. Denoting the state set of the actuators by $\mathscr{A}$,
we can again consider a channel from $\Sscr$ to $\mathscr{A}$ as formal model of
a policy, which we denote by $\pi$. Being more precise, with a sensor state $s$
and a subset $\mathsf{A}$ of actuator states, $\pi(s; \mathsf{A})$ denotes the probability that
the agent chooses an actuator state in $\mathsf{A}$, given that its sensor state is $s$.
Again, we have a Markov kernel \[ \pi: \; \Sscr \; \longrightarrow \;
\Delta_{\Ascr} \, , \] where $\Delta_{\Ascr}$ denotes the set of probability
distributions on $\Ascr$. We also use the notation $\pi(s; da)$ for an
infinitesimal set $da$ as we already introduced above in the context of $\beta$.
Note that this definition of a policy allows us to also consider a random choice
of actions, so-called non-deterministic policies. The set of policies is denoted
by $\Delta^{\Sscr}_{\Ascr}$.    \\

Finally, we consider the change of the world state from $w$ to $w'$ in the
context of an actuator state $a$ as a channel which we 
denote by $\alpha$. More precisely, given 
a world state $w$, an actuator state $a$, and a set $\mathsf{W'}$ of world states,
$\alpha(w,a; \mathsf{W'})$ denotes the probability that the actuator 
state $a$ will generate a transition from $w$ to a new world state that is in
$\mathsf{W'}$. As for the other channels, we use also in this case the notation  
$\alpha(w,a; dw')$ for infinitesimally small sets $dw'$. With the set $\Delta_{\Wscr}$ of probability distributions 
on $\Wscr$, we have 
 \[
    \alpha: \; \Wscr \times \Ascr \; \longrightarrow \;  \Delta_{\Wscr} \, .
\]   
We refer to $\alpha$ as {\em world channel\/} and denote the set of all world channels by $\Delta^{\Wscr \times \Ascr}_{\Wscr}$. \\

We have defined three mechanisms that are involved in a (reactive) sensorimotor loop of an embodied agent. Clearly, the agent's embodiment poses constraints to this loop, which we attribute to the mechanisms $\beta$ and $\alpha$. The agent is equipped with these mechanisms, but they are both considered to be determined and not modifiable by the agent. On the other hand, the 
policy $\pi$ can be modified by the agent in terms of learning processes.  
In order to describe the process of interaction of the agent with the world, we have to sequentially apply the individual mechanisms in the right order. Starting with an initial world state $w^t$ at time $t$, first the sensor state $s^t$ is generated in terms of the channel $\beta$. Then, based on the state of the sensor, an actuator state $a^t$ is chosen according to the policy $\pi$. Finally, the world makes a transition, 
governed by $\alpha$, from the state $w$ to a new state $w^{t + 1}$, which is influenced by the actuator state $a^t$ of the agent. 
Altogether, this defines the combined mechanism
\begin{equation} \label{integrated}
     {\Bbb P}^\pi(w^t; ds^t, da^t,dw^{t+ 1}) \, := \, \beta(w^t; ds^t) \, \pi(s^t; da^t) \, \alpha(w^t,a^t; dw^{t + 1}) \, .
\end{equation}
Note that we consider $\beta$ and $\alpha$ fixed and therefore emphasise only the dependence on $\pi$.
Now, with the new state $w^{t + 1}$ of the world, the three steps are iterated. 
This generates a process which is shown in Figure \ref{figure:SML}.   
Formally, the process is a probability distribution over trajectories that start with $w^0$:
 \begin{equation} \label{sequence}
      w^0, \quad s^0, a^0, w^1,\quad  s^1, a^1, w^2, \quad s^2, a^2, w^3, \quad \dots, \quad  s^{T-1}, a^{T-1}, w^{T} \, .
 \end{equation}
In order to describe this probability distribution, we have to iterate the mechanism (\ref{integrated}) by multiplication:
\[
   {\Bbb P}^\pi(w^0; ds^0, da^0, dw^1, \dots, ds^{T-1},da^{T-1}, dw^{T}) \; := \; \prod_{t = 0}^{T-1} {\Bbb P}^\pi(w^t; ds^t, da^t,dw^{t+ 1}) \, .
\]

\begin{figure}[h]
	\centering
	\includegraphics[clip=true,trim=0cm .5cm 0cm 0cm, scale=1.14]{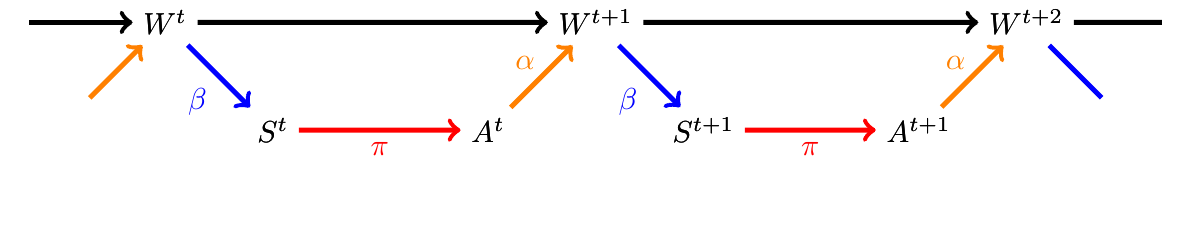}
	\caption{Causal structure of the reactive SML. The gray nuance groups the variables in one time step. } 
	\label{figure:SML}
\end{figure}

Now, what aspects of the sequence (\ref{sequence}) represent the behavior of the agent? Let us consider, for instance, a walking behavior. It is given as a movement of the agent's body in physical space, which is completely determined by the world process. Remember that the body is part of the world. Clearly, the particular sequence of sensor and actuator states 
does not matter as long as they contribute to the generation of the same body movement. Therefore, we consider the world process 
$w^t$ as the one in which behavior takes place and integrate out the other processes:
\begin{eqnarray}
     {\Bbb P}^\pi(w^0; dw^1, \dots, dw^{T})  & = &        
                       \underbrace{\int_{\Ascr} \int_{\Sscr}  \cdots  \int_{\Ascr} \int_{\Sscr}}_{\mbox{\small $T$ times}}  
                        {\Bbb P}^\pi(w^0; ds^0, da^0, dw^1, \dots, ds^{T-1},da^{T-1}, dw^{T}) \nonumber \\
                        & = &  \prod_{t = 0}^{T-1} {\Bbb P}^\pi(w^t; dw^{t+ 1}) \, .           \label{behavior}
\end{eqnarray}  
One can show that, with weak assumptions, the limit for $T \to \infty$ exists, so that we can write
\[    
   {\Bbb P}^\pi(w^0; dw^1, dw^2, \dots) \,,
\]   
which is a Markov kernel from an initial world state $w^0$ to the space of all infinite future sequences $w^1,w^2, \dots$. We denote the 
set of these Markov kernels by $\Delta^{\Wscr}_{\Wscr^\infty}$. This allows us to formalize the map that assigns to each policy the 
corresponding behavior:
\begin{equation}
    \psi_\infty: \Delta^{\Sscr}_{\Ascr} \;\; \longrightarrow \;\; \Delta^{\Wscr}_{\Wscr^\infty}, \qquad \pi \;\; \longmapsto \;\; {\Bbb P}^\pi(w^0; dw^1, dw^2, \dots)\,. 
\end{equation} 
We refer to this map as the {\em policy-behavior map\/}. 
Two policies $\pi_1$ and $\pi_2$ will be considered equivalent, if they generate the same behavior, that is, 
\begin{equation} \label{equivalent}
    \psi_\infty({\pi_1}) \; = \;  \psi_\infty({\pi_2}) \, .
\end{equation}
We argue that embodiment constraints render many equivalent policies. 
We can exploit this fact in order to design a concise control architecture. 
This will lead to a quantitative treatment of the notion of cheap design within the field of embodied intelligence. 
Let us treat this systems design problem in a more rigorous way.
As we pointed out, the agent is equipped with the mechanisms $\beta$ and $\alpha$ which constitute the embodiment of the agent. 
 In a biological system these mechanisms will change due to developmental processes. 
 However, we want to restrict our attention to the learning processes and disentangle them from developmental processes by assuming that the latter ones have already converged and therefore consider them as fixed. 
 Learning refers to a process in which the policy is changing in time. 
 Clearly, in order to model this change the agent has to be equipped with a family of possible policies, which we denote by ${\mathcal M}$, and refer to as {\em policy model}. 
For instance, we can consider neural networks as policy models that are parametrized by synaptic weights and threshold values for the individual neurons.  
Changing the weights and the thresholds 
will lead to a change of the policy (although there may be degeneracies, in general). 
In any case, going through all the possible parameter values will generate a set ${\mathcal M}$ of policies with which the agent is equipped for its behavior.

We argue that, if the embodiment constraints lead to many equivalent policies, then it is possible to find a concise model ${\mathcal M}$ that is capable of generating all behaviors. More precisely, we will require from a model
\[
      \psi_\infty({\mathcal M}) \; = \; \psi_\infty(\Delta^{\Sscr}_{\Ascr} ) \, ,
\]            
or, more precisely, a slight modification by taking limit points of ${\mathcal M}$ into account. 
We refer to this property of ${\mathcal M}$ as being an {\em embodied universal approximator}. In order to highlight the exploitation of embodiment constraints for cheap design, we compare this kind of universal approximation to the standard notion of universal approximation, which we refer to as {\em non-embodied universal approximation\/}.   \\

\section{Cheap Representation of Embodied Behaviors}
\label{section:embodied}

Intuitively it is clear that the embodiment constraints cause restrictions in the set of behaviors that an agent can realize. 
For example, inertia restricts the pace at which an embodied system can change its direction of motion (imagine a train switching the traveling direction instantaneously). 
In turn, not all world-state transitions may be possible in a single time step, regardless of what the policy specifies as a desirable action to take. 
These restrictions create a bottleneck between the set of policies on the one side and the set of possible behaviors on the other. 
The consequence is that, generically, infinitely many policies parametrize the same behavior. 
If we understand the way in which different policies are mapped to the same, or to different, behaviors, 
then we can parametrize all the behaviors that can possibly emerge in the SML by a low-dimensional (or low-complexity) set of policies. 
We develop the necessary tools in this section. 
For clarity we will focus on the reactive SML with finite sensor and actuator state spaces but allowing the possibility of a continuous world state. 
In particular we will use $\beta(w;s)$ instead of $\beta(w;ds)$ and $\pi(s;a)$ instead of $\pi(s;da)$. 
Possible generalizations of these settings are discussed in Appendix~\ref{sec:generalizations}. \\

The condition (\ref{equivalent}) is clearly the same as
\[
{\Bbb P}^{\pi_1}(w^0; dw^1, \dots, dw^{T}) \; = \;  {\Bbb P}^{\pi_2}(w^0; dw^1, \dots, dw^{T}), \qquad \mbox{for all $T = 1,2,\dots$}, 
\]
and with equation (\ref{behavior}), this is satisfied if and only if 
\[
{\Bbb P}^{\pi_1}(w; dw') \; = \;   {\Bbb P}^{\pi_2}(w; dw') \, . 
\]
Therefore, the mechanism ${\Bbb P}^\pi(w; dw')$ will play an important role in our analysis, and we consider the one-step formulation of the 
policy-behavior map:
\begin{equation}
\psi: \Delta^{\Sscr}_{\Ascr} \;\; \longrightarrow \;\; \Delta^{\Wscr}_{\Wscr}, \qquad \pi \;\; \longmapsto \;\; {\Bbb P}^\pi(w; dw' )\,,  
\end{equation} 
where 
\begin{equation}
\mathbb{P}^\pi(w;dw') = \sum_{s\in\Sscr}\sum_{a\in\Ascr}\beta(w;s) \pi(s;a) \alpha(a,w; dw'). \label{eq:realizable}
\end{equation}
This is an affine map from the convex set $\Delta^\Sscr_\Ascr$ to the convex set $\Delta^\Wscr_\Wscr$. 
The image of this map represents the set of all possible behaviors that the SML can generate. 
We denote this set by $\Beh$ and refer to its dimension the {\em embodied behavior dimension} $d = \dim(\psi(\Delta^\Sscr_\Ascr))$. 
This is given by the number linearly independent vectors 
\begin{equation}
\beta(w;s)(\alpha(w,a_0; dw') -\alpha(w,a;dw'))
\end{equation}
in $\aff(\psi(\Delta^\Sscr_\Ascr))$, where $(s,a)$ are sensor-actuator states with $a\neq a_0$, for some fixed $a_0\in\Ascr$. 
These vectors are namely the images of a basis of $\aff(\Delta^\Sscr_\Ascr)$. 
See Appendix~\ref{app:proofs} for more details about this. 
This allows us to formulate a simple upper bound for the embodied behavior dimension $d$ in terms of the image dimension of the maps $\beta$ and $\alpha$. 
Writing $\rank$ for the image dimension and regarding $\beta$ as a linear map $\mathbb{R}^\Sscr \to \mathbb{R}^\Wscr$ (operating on the columns of $\pi$) and $\alpha$ as an affine map $\Delta_\Ascr \to \Delta^\Wscr_\Wscr$ (operating on the rows of $\pi$), 
we have  
\begin{equation}
d =\rank(\psi) \leq \rank(\beta) \cdot \rank(\alpha). 
\label{eq:upperbounddim}
\end{equation}
For example, if $\Wscr$ is finite, 
$\rank (\beta)$ is the rank of the matrix with entries $(\beta(w;s))_{w\in\Wscr, s\in\Sscr}$ and $\rank (\alpha)$ is the rank of the matrix with entries $(\alpha(w,a_0;w') -\alpha(w,a;w') )_{a\in\Ascr, (w,w')\in\Wscr\times\Wscr}$, for some fixed $a_0\in\Ascr$. 
Even though the upper bound~\eqref{eq:upperbounddim} does not necessarily hold with equality (the rank of the combined map $\psi$, where $\beta$ and $\alpha$ share the same $w$, can be significantly smaller than the product of the individual ranks), it gives us a clear picture of how the embodiment constraints, represented by $\beta$ and $\alpha$, can lead to an embodied behavior dimension $d$ that is much smaller than $|\Sscr|(|\Ascr|-1)$, the dimension of the set of all policies $\Delta^\Sscr_\Ascr$.

\begin{figure}
	\centering
	\includegraphics[width=7cm]{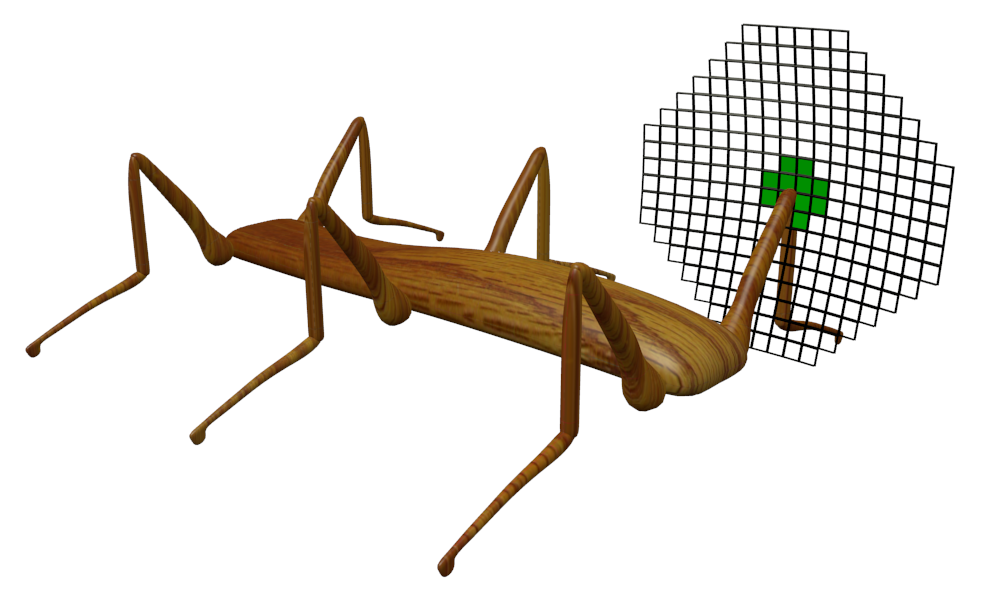}
	\caption{Locality of world-state transitions. This hexapod will be used in the experimental evaluation of our theory in Section~\ref{section:experiments}. }
	\label{figure:grid}
\end{figure}

In an embodied system the sensors are usually insensitive to a large number of variations of the world-state $w$ and, therefore, $\beta$ is piece-wise constant with respect to $w$. 
Also, the sensors implement a certain degree of redundancy, 
meaning that, for each $w$, the probability distribution $\beta(w;\cdot)\in \Delta_\Sscr$ has certain types of symmetries. 
This means that $\rank(\beta)$ is much smaller than $|\Sscr|$ (the maximum theoretically possible rank). 
In the case of $\alpha$, 
usually several actions $a$ produce the same world-state transition, 
such that, for any fixed world state $w$, $\alpha(w,\cdot;\cdot)$ is piece-wise constant with respect to $a$. 
Furthermore, for any given $w$, only very few states $w' \in \Wscr$ are possible at the next time step, 
regardless of $a$, such that $\alpha(w,a;\cdot)$ assigns positive probability only to a very small subset of $\Wscr$. 
This means that $\rank(\alpha)$ is usually much smaller than $(|\Ascr|-1)$ (the maximum theoretically possible rank). 
An example for this kind of constraints on $\alpha$ is a robot's knee, which in a time step can only be moved to adjacent positions, as the one shown in Figure~\ref{figure:grid}. 
\\

So far we have discussed the embodied behavior dimension of an embodied system $d=\dim(\Beh)$ 
and why this can be much smaller than the dimension of the policy space $\dim(\Delta^\Sscr_\Ascr)=|\Sscr|(|\Ascr|-1)$. 
Since the policy-behavior map $\psi$ is affine, for any generic behavior that can possibly emerge in the SML, there is a $(|\Sscr|(|\Ascr|-1) -d)$-dimensional set (in fact a polytope) of equivalent policies generating that same behavior. 
By selecting representatives from each set of equivalent policies, 
we can define low-dimensional policy models which are just as expressive as the much higher dimensional set $\Delta^\Sscr_\Ascr$ of all possible policies, in terms of the representable behaviors. 
The following example shows that it is possible to define a smooth manifold of policies which 
translate in a one-to-one fashion to the set of all possible behaviors in the SML. 

\begin{example} 
	\label{example:mineua}
	Consider the matrix $E\in\R^{d\times (\Sscr\times\Ascr)}$ that represents the policy-behavior map $\psi$ with respect to some basis. 
	Then the exponential family $\Ecal^\Sscr_\Ascr$ of policies defined by 
	\begin{equation}
		\pi_\theta(s;a) = \frac{\exp( \theta^\top E (s,a)  )}{\sum_{a'\in\Ascr} \exp( \theta^\top E (s,a')  )} ,\quad \text{for all $a\in\Ascr$ and $s\in\Sscr$},\quad \text{for all $\theta\in\R^d$},\label{eq:mineua}
	\end{equation}
	is an embodied universal approximator of dimension $d$. 
	In fact, each behavior from the set $\Beh$ is realized by exactly one policy from the set $\overline{\Ecal^\Sscr_\Ascr}$.  
	See Figure~\ref{fig:expfamcond} for an illustration and Appendix~\ref{app:proofs} for technical details. 
\end{example}

\begin{figure}
	\centering
\includegraphics[clip=true, trim=.73cm 0cm .47cm 0cm, scale=1]{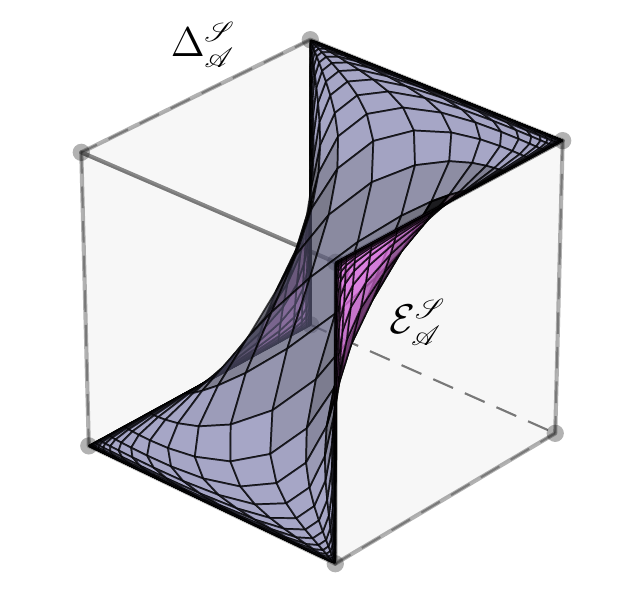}
	\caption{
		Illustration of the exponential family of policies described in Example~\ref{example:mineua}. 
		Here $|\Sscr|=3$ and $|\Ascr|=2$, such that the polytope $\Delta_{\Ascr}^{\Sscr}$ is the three-dimensional cube of $3\times 2$ row stochastic matrices. 
		The curved surface is the exponential family $\Ecal^\Sscr_\Ascr$ from the example for some specific $\beta$ and $\alpha$ with embodied behavior dimension $d=2$. 
		In the SML defined by $\beta$ and $\alpha$, the policies from $\Ecal^\Sscr_\Ascr$ generate the same behaviors as the set of all policies~$\Delta_\Ascr^\Sscr$.  \label{fig:expfamcond}
	}
\end{figure}

The previous discussion shows that the set of behaviors that can possibly emerge in the SML 
usually has a much lower dimension than the set of all possible policies. Furthermore, it shows that it is possible to construct low-dimensional embodied universal approximators. 
Nonetheless, among all behaviors that are possible in the SML, 
we can expect that only a smaller subset $\Bcal\subset\Beh$ is actually relevant to the agent. 
For instance, among all locomotion gaits that an agent could possibly realize with its body in a given environment, 
we can expect that it will only 
utilize those which are most successful (with respect to different tasks). 
The embodied behavior dimension can be directly generalized in order to capture such behavioral restrictions, in addition to the embodiment constraints $\alpha$ and $\beta$. 
We note that specifying a set $\Bcal$ of potentially interesting behaviors is in general not an easy task.  
For instance, one could be interested in behaviors that maximize the expected reward of an agent (besides from implicit definitions, say as the maximizers of a given objective function). 
As the main focus of this paper is not to find a good definition of interesting or natural behaviors, 
here we will consider a very simple but general classification of behaviors, in terms of their supports. 

In relation to the information self-structuring mentioned in the introduction, for specific behaviors usually only a relatively small subset $\Scal\subseteq\Sscr$ of sensor values emerges. 
In such situations, the policy only needs to be specified for the sensor values in the subset $\Scal$. 
\newcommand{\Wcal}{\mathcal{W}}
Consider a set of behaviors $\Bcal$ that take place within a restricted set of world states 
$\Wcal \subseteq \Wscr$ and consider the set of sensor states that can be possibly measured from these world states,  
$\Scal := \{s\in \Sscr \colon s \in \supp (\beta(w; \cdot)) \text{ for some $w\in \Wcal$} \}$. 
For the  world states $\Wcal$, the measurement by $\beta$ always produces sensor values in $\Scal$, 
and the policy for states not in $\Scal$ does not play any role. 
We denote by $\psi^\Scal$ the restriction of the policy-behavior map to $w\in \Wcal$. 
This is given by the natural restriction of the kernels $\beta$ and $\alpha$ to the domain $\Wcal\subseteq\Wscr$. 
In this case, the embodied behavior dimension is given by $d^{\Scal}  := \dim(\psi^\Scal(\Delta^\Sscr_\Ascr))$. 
We will denote a model $\Mcal\subseteq\Delta_\Ascr^\Sscr$ an embodied universal approximator on $\Scal$ if 
$\psi^\Scal(\overline{\Mcal}) = \psi^\Scal(\Delta^\Sscr_\Ascr)$. 
This definition means that the model is powerful enough to control any behavior on $\Wcal$ just as well as the entire policy polytope. 
Given any set of behaviors $\Bcal$, e.g., as the one described above, we are interested in the following problem. 

\begin{problem}
	\label{problem}
	For a given set of possible behaviors $\Bcal\subseteq\Beh = \psi(\Delta^\Sscr_\Ascr)$ and a class of policy models $\Mfrak$, 
	what is the smallest model $\Mcal\in\Mfrak$ that can generate all these behaviors, such that  $\Bcal\subseteq \psi(\overline{\Mcal})$? 
\end{problem}

Of particular interest are classes of policy models $\Mfrak$ defined in terms of neural networks. 
In Section~\ref{section:theory} we will consider a class of policy models defined in terms of CRBMs. 
The following result gives us a simple and powerful combinatorial tool for addressing Problem~\ref{problem}. 

\begin{lemma} 
	\label{proposition:euasupp}
	Any model $\Mcal\subseteq\Delta_\Ascr^\Sscr$ with the following property is an embodied universal approximator on $\Scal$: 
	for every policy $\pi\in\Delta_\Ascr^\Sscr$ whose $\Scal$-rows have a total of $|\Scal|+d^{\Scal}$ or less non-zero entries, 
	there exists a policy $\pi^\ast\in\overline{\Mcal}$ with $\pi(s;\cdot)=\pi^\ast(s;\cdot)$ for all $s\in\Scal$. 
\end{lemma}

This lemma states that for universal approximation of embodied behaviors it suffices to approximate the policies which, 
	for a relevant set of sensor values, assign positive probability only to a limited number of actions. The number of actions is determined by the embodied behavior dimension. 

It is worthwhile mentioning that Example~\ref{example:mineua} and Lemma~\ref{proposition:euasupp} describe two complementary types of universal approximators of embodied behaviors. 
The first type, described in the example, is composed of maximum entropy policies, whereas the second type, described in the lemma, is composed of minimum entropy policies. 
If we consider the set of equivalent policies that map to a given behavior, 
Example~\ref{example:mineua} selects the one with the most random state-action assignments that are possible for generating that behavior. 
On the other hand, Lemma~\ref{proposition:euasupp} selects the ones with the most deterministic state-action assignments that are possible for generating that behavior. 
Geometrically, the set of equivalent policies of a given behavior is the convex hull of the minimum entropy policies, with the maximum entropy policy lying in the center. 
The exponential family has nice geometric properties, but it is very specific to the kernels $\beta$ and $\alpha$, which define the sufficient statistics. 
The set described in Lemma~\ref{proposition:euasupp} can also be considered as a policy model. It offers several advantages that we will exploit latter on. 
First, it has a very simple combinatorial description. Second, it only depends on the embodied behavior dimension $d$, irrespective of the specific kernels $\beta$ and $\alpha$ (which are not directly accessible to the agent). Third, it selects policies with the minimum possible number of positive-probability actions, 
which seems natural from a  concise controller.

\section{A Case Study with Conditional Restricted Boltzmann Machines}\label{section:univ_sub}\label{section:theory}

\subsection{Definitions}

A Boltzmann machine (BM) is an undirected stochastic network with binary units, some of which may be hidden. 
It defines probabilities for the joint states of its visible units, given by the relative frequencies at which these states are observed, asymptotically, depending on the network parameters (interaction weights and biases). 
The probability of each joint state $x=(x_V,x_H)$ of the network's visible and hidden units 
can be described by the Gibbs-Boltzmann distribution 
$p(x)=\frac{1}{Z} \exp(-\Hcal(x))$  
with energy function $\Hcal(x)=\sum_{i,j}x_i W_{ij} x_j +\sum_i b_i x_i$ and normalization partition function $Z(W,b)=\sum_{x'}\exp(-\Hcal(x'))$. 
The probabilities of the visible states are given by marginalizing out the states of the hidden units, 
$p(x_V)=\sum_{x_H}p(x_V,x_H)$.

\begin{figure}
\centering
\includegraphics[scale=1.1]{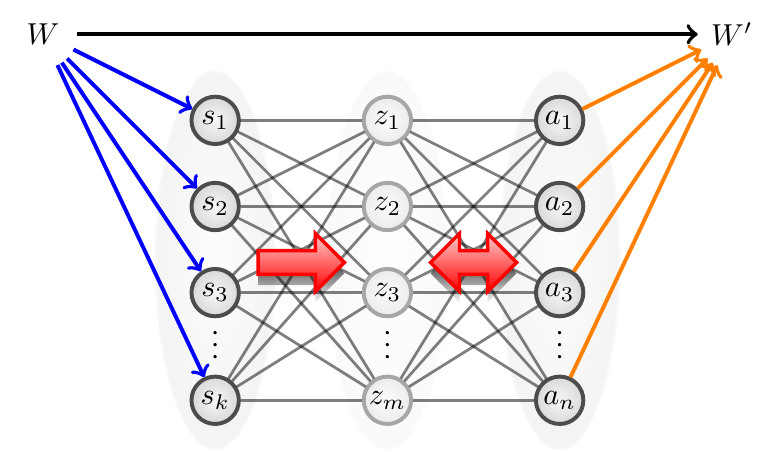}
	\caption{Illustration of a CRBM policy in the SML. 
		}
\label{fig:crbm}
\end{figure}

An RBM is a BM with the restriction that there are no interactions between the visible units nor between the hidden units, such that $W_{ij}\neq 0$ only when unit $i$ is visible and $j$ hidden. 
As any multivariate model of probability distributions, 
RBMs define models of conditional distributions, given by clamping the state of some of the visible units: 

\begin{definition}\label{definitionCRBM}
The conditional restricted Boltzmann machine model with $k$ input, $n$ output, and $m$ hidden units, denoted $\RBM_{n,m}^k$, is the set of all conditional distributions in 
$\Delta^\Sscr_\Ascr$, $\Sscr=\{0,1\}^k$, $\Ascr=\{0,1\}^n$,  
that can be written as 
\begin{multline*}
p(x|y) =  \frac{1}{Z(W,b,Vy+c)} \sum_{z\in\{0,1\}^m} \exp(z^\top W x + z^\top V y + b^\top x 
+ c^\top z),   
\quad \forall x\in\{0,1\}^n, y\in\{0,1\}^k, 
\label{definitioneq} 
\end{multline*}
where $Z(W,b,Vy+c)$ normalizes the probability distribution $p(\cdot|y)\in\Delta_\Ascr$, for each $y\in\{0,1\}^k$. 
Here, $y$, $x$, and $z$ are state vectors of the input, output, and hidden units, respectively. 
Furthermore, $V\in\R^{m\times k}$ is a matrix of interaction weights between hidden and input units, $W\in\R^{m\times n}$ is a matrix of interaction weights between hidden and output units, $c\in\R^m$ is a vector of biases for the hidden units, and $b\in\R^n$ is a vector of biases for the output units. Here $\top$ denotes vector transposition. 
\end{definition}

The model $\RBM_{n,m}^k$ has $m k + m n + m + n$ parameters (the interaction weights and biases). 
A bias term $a^\top y$ for the input units does not appear in the definition, as it would cancel out with the normalization function $Z(W,b,Vy+c)$. 
When there are no input units, i.e., $k=0$, the conditional probability model $\RBM_{n,m}^k$ reduces to the restricted Boltzmann machine probability model with $n$ visible and $m$ hidden units, which we denote by $\RBM_{n,m}$.

\subsection{Non-Embodied Universal Approximation}

In this section we ask for the minimal number of hidden units $m$ for which the model $\RBM_{n,m}^k$ can approximate every conditional distribution from the set $\Delta^\Sscr_\Ascr$ with $\Sscr =\{0,1\}^k$ and $\Ascr = \{0,1\}^n$, denoted  $\Delta^k_n$, arbitrarily well. 
Later we will contrast this non-embodied universal approximation with the embodied case. 

Note that each conditional distribution $p(x|y)$ can be identified with the set of joint distributions of the form $r(x,y)=q(y) p(x|y)$, with strictly positive marginals $q(y)$. 
By fixing $q(y)$ equal to the uniform distribution over $\Sscr$, 
we obtain an identification of $\Delta_\Ascr^\Sscr$ with $\frac{1}{|\Sscr|}\Delta_{\Ascr}^{\Sscr} \subseteq \Delta_{\Sscr\times\Ascr}$. 
In particular, we have that universal approximators of joint probability distributions define universal approximators of conditional distributions. 
This observation allows us to translate results on the representational power of RBMs to corresponding results for CRBMs. 
For example, we know that $\RBM_{k+n,m}$ is a universal approximator of probability distributions on $\{0,1\}^{k+n}$ whenever $m\geq \frac{1}{2} 2^{k+n} -1$~\citep[see][]{Montufar2011}, and therefore:  
\begin{proposition}
\label{proposition:universal}
The model $\RBM_{n,m}^k$ can approximate every conditional distribution from $\Delta_{n}^{k}$ 
arbitrarily well whenever $m\geq \frac{1}{2}2^{k+n}-1$. 
\end{proposition}

Now, since conditional models do not need to model the input-state distributions, 
in principle it is possible that $\RBM_{n,m}^k$ is a universal approximator of conditional distributions even if $\RBM_{n+k,m}$ is not a universal approximator of probability distributions. 
Therefore, we also consider a result by~\cite{montufar2014CRBMs}, which improves Proposition~\ref{proposition:universal} and does not follow from corresponding results for RBM probability models: 
\begin{theorem}
\label{theorem:universal}
The model $\RBM_{n,m}^k$ can approximate every conditional distribution from $\Delta_{n}^{k}$ 
arbitrarily well whenever 
$m\geq \frac{1}{2} 2^{k} (2^n-1) = \frac{1}{2} 2^{k+n} - \frac{1}{2} 2^k$.  
\end{theorem}
 
Theorem~\ref{theorem:universal} represents a substantial improvement of Proposition~\ref{proposition:universal} in that it reflects the structure of $\Delta_n^k$ as a $2^k$-fold Cartesian product of the $(2^n-1)$-dimensional probability simplex $\Delta_n$, in contrast to the proposition's bound, which rather reflects the structure of the $(2^{k+n}-1)$-dimensional joint probability simplex~$\Delta_{k+n}$. 
The full statement of the theorem is quite technical, and thus we refer the interested reader to~\citep{montufar2014CRBMs}. 
At this point let it suffice to say that some terms appearing in the bound on $m$ decrease with increasing $k$, 
such that approximately the prefactor $\tfrac{1}{2}$ decreases to $\frac{1}{4}$ when $k$ is large enough. 

As expected, the asymptotic behavior of this result is exponential in the number of input and output units. 
We believe that the result is reasonably tight, although some improvements may still be possible. 
A crude lower bound can be obtained by comparing the number of parameters with the dimension of the policy polytope~\cite[see details in][]{montufar2014CRBMs}: 

\begin{proposition}
\label{proposition:universallower}
If the model $\RBM_{n,m}^k$ can approximate every policy from $\Delta_{n}^{k}$ 
arbitrarily well, then necessarily 
$m \geq \frac{1}{{(n+k+1 )}}(2^k(2^n-1) -n)$. 
\end{proposition}

\subsection{Embodied Universal Approximation}

By Theorem~\ref{proposition:euasupp}, 
we can achieve embodied universal approximation by approximating only policies with a limited number of non-zero entries. 
Furthermore, as mentioned earlier, if we only care about a subset $\Scal \subseteq\Sscr$ of sensor values, 
we can restrict the policy space $\Delta_\Ascr^\Sscr$ to $\Delta_\Ascr^\Scal$. 
This means that the number of policy entries that we need to model is given by the number of interesting sensor values plus the corresponding embodied behavior dimension. 
On the other hand, we can use each hidden unit of a CRBM to model each relevant non-zero entry of the policy. 

\begin{theorem}
\label{thm:restricted_embodied_universal_approximation}
The model $\RBM_{n,m}^k$ is an embodied universal approximator on $\Scal$ whenever $m \geq |\Scal| + d^{\Scal}-1$.
\end{theorem}
\begin{proof}
	We use Lemma~\ref{proposition:euasupp}. 
	The joint probability model $\RBM_{k+n,m}$ can approximate any probability distribution with support of cardinality $m+1$ arbitrarily well~\citep[see][]{Montufar2011}. 
	Hence, with $m\geq |\Scal|+d^\Scal-1$, $\RBM_{n+k,m}$ can approximate any joint distribution with $|\Scal|+d^{\Scal}$ non-zero entries within the set $\{(s,a)\colon s\in\Scal, a\in\Ascr \}$ arbitrarily well. 
	These joint distributions represent a set $\Mcal$ of conditional distributions of $a$ given $s$ such that, 
	for any policy $\pi$ whose $\Scal$-rows have a total of $|\Scal|+d^\Scal$ or less non-zero entries, 
	there is a conditional distribution $\pi^\ast\in\Mcal$ with $\pi^\ast(\cdot|s) = \pi(\cdot|s)$ for all $s\in\Scal$. 
\end{proof}

This theorem gives a bound for the number of hidden units of CRBMs that suffices to obtain embodied universal approximation. The bound depends on the embodiment and behavioral constraints of the system, captured in the embodied behavior dimension $d^\Scal$. 
In general, this bound will be much smaller than the exponential bound from Theorem~\ref{theorem:universal}. 
We will test this result in the context of particular behavioral constraint on a hexapod in the next section.

\section{Experiments with a Hexapod}
\label{section:experiments}

In the previous sections we have derived a theoretical bound for the complexity of a CRBM based policy. In this section, we want to evaluate that bound experimentally. 
For this purpose, we chose a six-legged walking machine (hexapod) as our experimental platform (see Figure~\ref{fig:hexapod}), because it is a well-studied morphology in the context of artificial intelligence, with one of
its first appearances as Ghengis~\citep{Brooks1989aA-Robot}. 
The purpose of this section is \emph{not} to develop an optimal walking strategy for this system. 
Contrary, this morphology was chosen, because the tripod gait (see Figure~\ref{fig:hexapod}) is known to be one of the optimal locomotion behaviors, which can be implemented efficiently in various ways. 
This said, learning a control for this gait 
is not trivial, and hence, a good test bed to evaluate our complexity bound for CRBM based policies.

This section is organized in three parts. The first part presents the experimental set-up as far as is it required to understand the results. The second part describes how the CRBM complexity parameter $m$ was estimated form the data. The third part presents the results of the experiment and compares them with the theoretical bound. 

\subsection{Simulation} 
\label{sec:simulation} 
The hexapod was simulated with \textsf{YARS}~\citep{Zahedi2008aYARS:}, which is a
mobile robot simulator based on the bullet physics engine~\citep{Coumans2012aBullet}. 
Each segment of the hexapod is defined by its
physical properties (dimension, weight, etc.) and each actuator is defined by
its force, velocity and its angular range. In the case of the hexapod shown in
Figure~\ref{fig:hexapod} (left-hand side), the main body's dimension (bounding box)
is 4.4m length, 0.7m width, 0.5m height, and the weight is 2kg. 
Each leg consists of three
segments (femur, tarsus, tibia), of which the two lower segments (tarus, tibia)
are connected by a fixed joint. The leg segments were freely modeled with
respect to the dimensions of an insect leg. The actuator which connects the
femur and tarsus (knee actuator) only allows rotations around the local $y$-axis
of the femur segment (see Figure~\ref{fig:hexapod}). The maximal deviation for
the femur-tarsus actuator is limited to $\omega_\mathrm{fe-ta}\in[-15^\circ,
25^\circ]$. For actuators which connect the main body with the femur (ma-fe),
the maximal deviation is limited to $\omega_\mathrm{ma-fe}\in[\pm35^\circ]$. The
rotation axis of the ma-fe actuator is limited to the local $z$-axis of the main
body. In bullet, an actuator is defined by its impulse (set to $1\mathrm{Ns}$
for both actuators) and its maximal velocity (set to
$0.75\mathrm{rad/s}\approx42^{\circ}/s$ for both actuators). It must be noted
here, that the sensors and actuators are all mapped onto the interval $[-1,1]$,
which means that a sensor value of $\tilde{S}_i^t = 1$ refers to the maximal
current deviation of the corresponding joint. In the same sense, an
actuator value $\tilde{A}_j^t=1$ refers to the motor command to deviate the
corresponding joint to its maximal position. For simplicity, we refer to the two types of joints as shoulder (main body-femur) and knee (femur-tarus).

\begin{figure}[t]
	\begin{center}
		\begin{minipage}[c]{0.35\textwidth}
			\includegraphics[width=\textwidth]{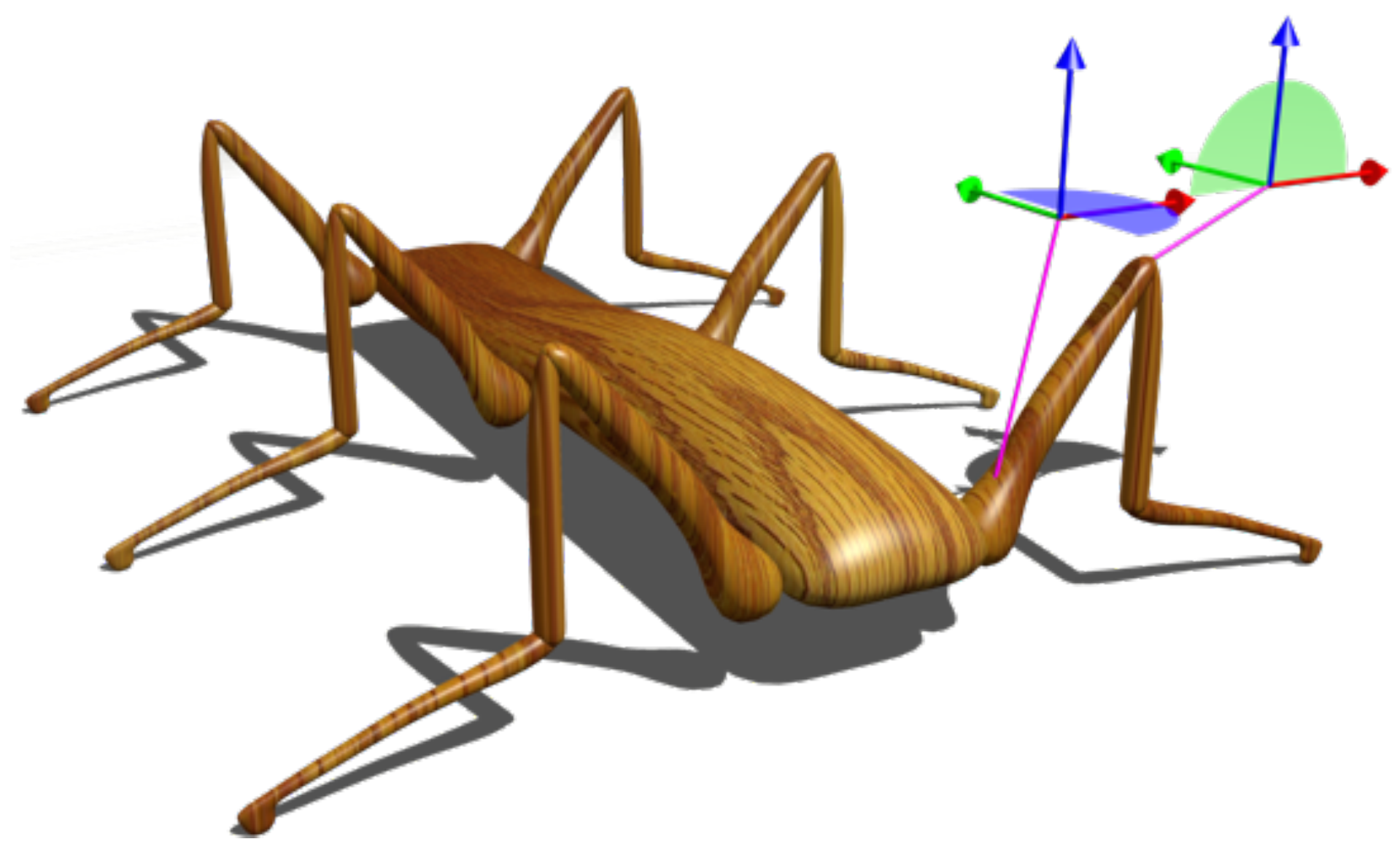}
		\end{minipage}
		\begin{minipage}[c]{0.64\textwidth}
			\includegraphics[width=\textwidth]{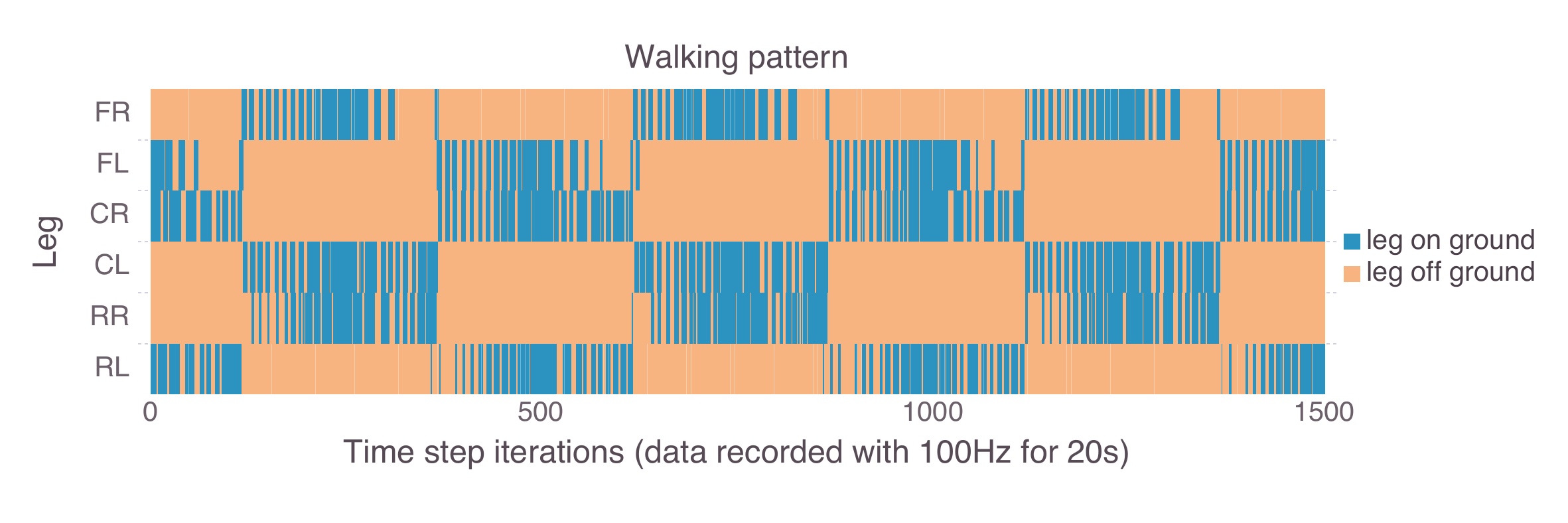}
		\end{minipage}
	\end{center}
	
	\caption{Hexapod set-up. Left-hand side: The simulated hexapod with a display of the joint configurations. For a detailed discussion, see Section~\ref{sec:simulation}. Right-hand side: Visualization of the target walking pattern. The plot shows which leg touched the ground at which point in time. Blue areas refer to a contact with a the ground, while orange areas refer to points in time during which the correspond leg did not touch the ground. The different legs are plotted over the $y$-axis, while each point on the $x$-axis refers to a single point in time.}
	\label{fig:hexapod}
\end{figure}

The policy update frequency was set to 10Hz, i.e.,~the controller received ten sensor values and generated ten actuator values per second. The target behavior of the hexapod (a tripod walking gait, see Figure~\ref{fig:hexapod}) was generated by an open-loop controller which applied phase shifted sinus oscillations to the actuators. For each leg, the sinus oscillations were discretized into 50 pairs of actuator values, which means that one  locomotion step requires 50 time steps (5 seconds) to complete. 

For the training and analysis, the sensor and actuator value data was
discretized into 16 equidistant bins for each sensor and actuator. This
corresponds to four binary input units for each sensor and four binary output
units for each actuator. Combined into two random variables
$S=(S_1, S_2, \ldots, S_{12}), A=(A_1, A_2, \ldots, A_{12})$,
 this leads to a total of
$16^{12}$ possible values ($|\Sscr| = |\Ascr| = 16^{12}$) corresponding to a total of $48$ binary input and $48$ binary output units. 
In the following sections, we only refer to this pre-processed data, which means that calculations and the training of the CRBMs described in the remainder of this section refer to the two random variables $S,A$. 

The next section will discuss the estimation of the sufficient controller complexity that is able to reproduce the desired tripod walking gait.

\subsection{Estimation of the Sufficient Complexity}
\label{sec:estimation of the required controller complexity}
Before the estimation procedure and results are presented, we restate the
inequality given in Theorem~\ref{thm:restricted_embodied_universal_approximation}, which is given by
\begin{eqnarray}
m & \geq |\Scal| + d^\Scal -1. \label{eq:m}
\end{eqnarray}
This means that a CRBM should not require more hidden units ($m$) than the sum of the support set cardinality $|\Scal|$ and embodied behavior dimension $d$ (minus 1). The following paragraphs explain how these two values were calculated from the recorded data.

The first step in estimating the sufficient controller complexity $m$ of the
CRBM policy model is the estimation of the support's cardinality $|\Scal|$. It
was mentioned above that there are $16^{12}$ possible sensor values. The
necessary complexity of a CRBM policy for a specific behavior depends on the
actually used number of sensor values, which is also known as the sensor support
set. By estimating the cardinality of the support set, we
know how many relevant sensor values the CRBM needs to handle to reproduce the
behavior of interest.

The estimation of the support set cardinality depends on the quality of the sample. 
Therefore, we sampled $10^5$ sensor values to ensure a sufficient convergence of the relative frequencies of the sensor values. 
The upper left plot in Figure~\ref{fig:hexapod results} shows the histogram for all recorded sensor values. 
The orange vertical line shows where we have pruned the data so that 80\% of the recorded data was kept. 
The lower left plot in Figure~\ref{fig:hexapod results} shows the remaining data. 
With this procedure (recording, estimating relative frequencies, pruning the
data to 80\%), we estimated the cardinality of the sensor support set at $|\Scal| = 63$. 
The pruning threshold of 80\% might appear arbitrary here. To clarify,
estimating the support from data is an interesting research topic by itself,
which, however, goes beyond the scope of this work. Our underlying assumption
for the pruning is that the sampling is noisy. We estimated the noise
empirically by analyzing the histogram and decided for a threshold that seemed
reasonable to us. We want to point out that the threshold was estimated before the
results of the experiments (see next section) were available.

The next step is to estimate the embodied behavior dimension, which is done here
based on the affine rank of the empirically estimated internal world model
$\gamma(s, a; s')$. For the sake of readability, we defer the justification for
the replacement of the embodiment-behavior dimension by the affine rank of the
internal world model to the appendix (see Appendix~\ref{sec:exp embodiment
dimension}). 

Given the internal world model $\gamma$, the affine rank is calculated in the following way:
\begin{equation}
d^{\Scal} = 
      \sum_{s\in\Scal} \rank( (\gamma(s,a_0 ; s') - \gamma(s, a ; s'))_{s'\in\Scal, a\in\Ascr} ) . \label{eq:affine rank}
\end{equation}
To estimate the internal world model $\gamma(s,a;s')$, we pruned the data in
accordance with the estimated support set cardinality. This means that we
removed all pairs of $S,A$ for which $S$ is not the in pruned support set
$\Scal$. For the remaining data, we counted the occurrences of $s^{t+1},s^t,a^t$
and filled the matrix $\gamma(s^t,a^t;s^{t+1})$. The matrix is initialized with
zero and each row is normalized by the row sum.  The resulting matrix does not
model a conditional probability distribution, 
because many rows have row sum zero. As we are only interested in the affine rank of the matrix, this is of no matter to us. The resulting value is $d^{\Scal} = 3$.

\paragraph{Resulting estimation of the model complexity:} It follows that the CRBM is able to represent the target behavior 
whenever the number of hidden units satisfies 
\begin{eqnarray}
m & \geq  |\Scal| + d^{\Scal} - 1 = 63 + 3 - 1= 65 .
\end{eqnarray}

\begin{figure}[t]
	\centering
	\begin{minipage}[c]{5cm}
		\includegraphics[width=\textwidth]{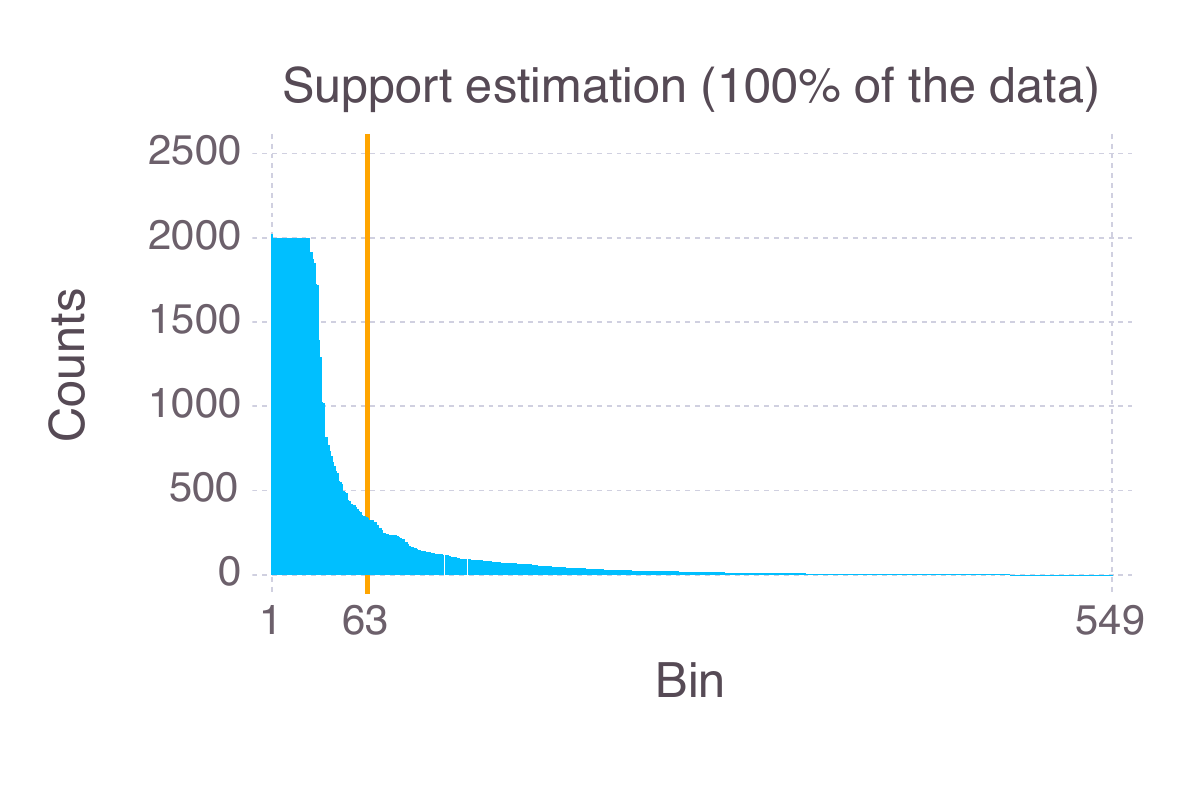}\\
		\includegraphics[width=\textwidth]{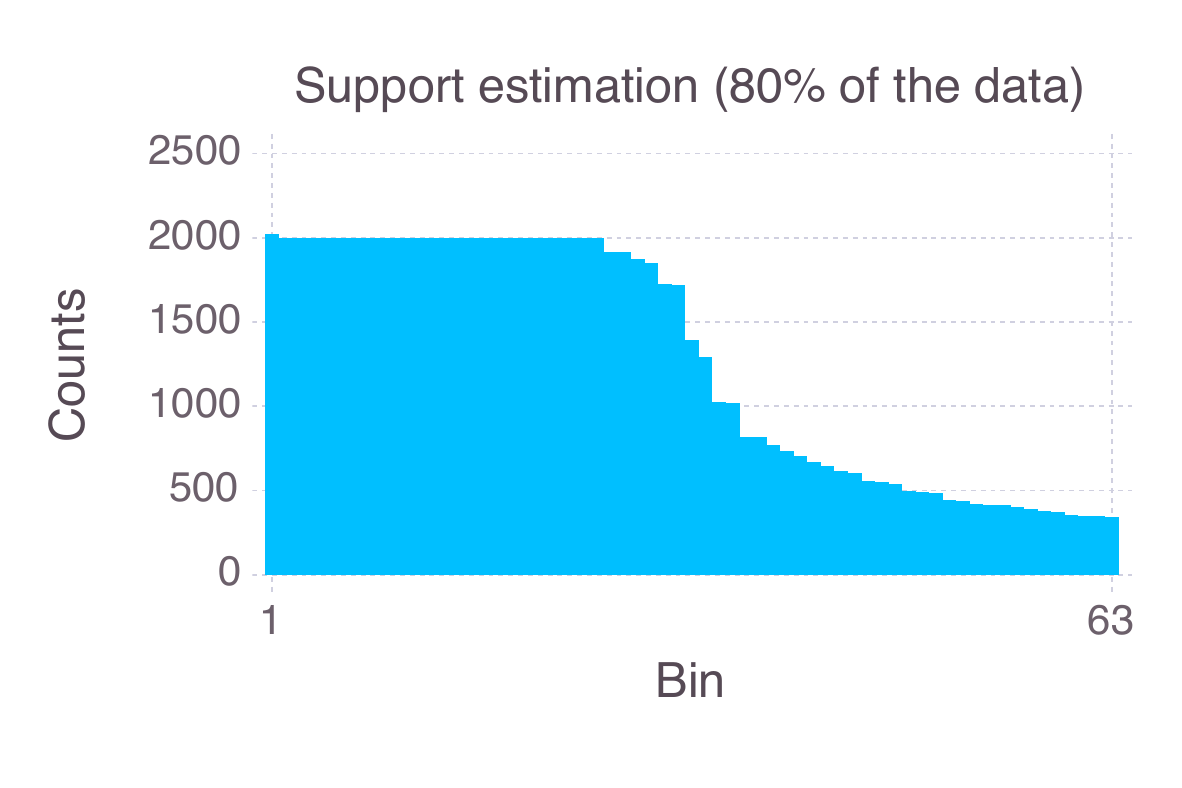}  
	\end{minipage}
	\begin{minipage}[c]{10cm}
		\includegraphics[width=\textwidth]{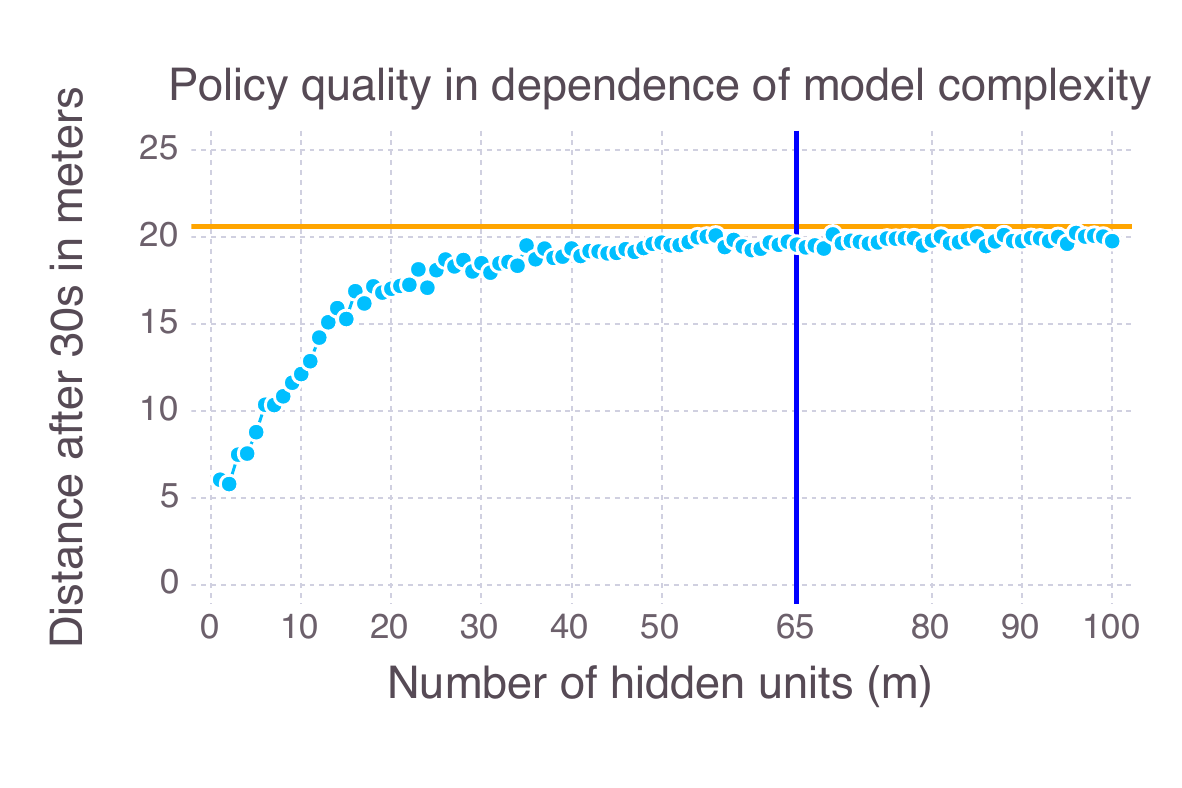}
	\end{minipage}
	\caption{Experimental results. Left-hand side: Estimation of the support set cardinality (before and after pruning). Right-hand side: Performance of the best CRBM for different complexity parameters $m$ in comparison to the performance of the target behavior (horizontal orange line). The vertical blue line indicates the $m$ estimated from the data (see Section~\ref{sec:estimation of the required controller complexity}).}
	\label{fig:hexapod results}
\end{figure}

To evaluate the tightness of this bound, we conducted a series of experiments, which are explained in the following section.
\subsection{Experiments to Evaluate the Tightness of the Complexity Estimation}
Before the experiments can be described, there is an important note to make.
This work is concerned with the minimally required complexity that is sufficient
to control an embodied agent, such that it is able to produce a set of
desirable behaviors (which also includes all behaviors as well as one specific
behavior). Here, we are \emph{not} concerned with the question how these
CRBMs should be trained optimally.
This is why we used a standard training algorithm for
RBMs~\citep{Hinton:2002:TPE:639729.639730,Hinton:practical} and conducted a
large scan over different complexity parameters $m$. For each
$m=1,2,3,\ldots,100$ we trained 100 CRBMs with the following learning
parameters: epochs = 20000, batch size = 50, learning rate $\alpha$ = 1.0,
momentum = 0.1, Gaussian distributed noise on sensor data = 0.01, weight cost = 0.001, bins = 16, CRBM update iterations = 10, on a data set of $10^4$ pairs of sensor and actuator values. 
Each trained CRBM was evaluated ten times, by applying it to the hexapod and
recording the covered distance for 30 seconds. The performance of the CRBMs is
measured against the target tripod walking gait, which achieves 20.6 meters
during the same time. As we are not concerned with the training of the CRBMs in
this work but instead with the best performance of a CRBM for a given $m$, we choose to take the maximally covered distance over all 1000 evaluations (100 CRBMs, 10 evaluation runs for each CRBM) to estimate the quality of a CRBM with a given complexity $m$. The plot on the right-hand side of Figure~\ref{fig:hexapod results} shows the resulting performance for all scanned values of $m$. The results show that our estimation is fairly tight, which means the performance of the CRBMs converges to the optimal behavior close to the estimated value of $m=65$.

\section{Conclusions}
\label{section:discussion}

We presented an approach for studying and implementing cheap design in the
context of embodied artificial intelligence. In this context, we referred to
cheap design as the reduction of the controller complexity that is possible
through an exploitation of the agent's body and environment. We developed a
theory to determine the minimal controller complexity that is sufficient to
generate a
given set of desired behaviors. Being more precise, we studied the way in which
embodiment constraints induce equivalent policies in the sense that they
generate the same observable behaviors. This led to the definition of the
effective dimension of an embodied system, the {\em embodied
behavior dimension}. In this way, we were able to define low-dimensional policy
models that can generate all possible behaviors. Such policy models are related
to the classical notion of universal approximation.

We used CRBMs as a platform of study, for which we presented non-trivial
universal approximation results in both the non-embodied and the embodied
settings. While the non-embodied universal approximation requires an enormous
number of hidden units (exponentially many in the number of input and output
units), embodied universal approximation can be achieved using essentially only
as many hidden units as the effective dimension of the system. Notably, our
construction depends only on the embodied behavior dimension and, therefore,
is independent of the specific embodiment constraints.

Experiments conducted on a walking machine demonstrate the tightness of
the estimated number of hidden units for a CRBM controller.
This shows the practical utility of our theoretical analysis.
To the best of our knowledge, the presented formalism and results 
are amongst the first quantitative contributions to cheap design in
embodied artificial intelligence.

\appendix

\subsection*{Appendix}

\section{Technical Proofs}
\label{app:proofs}
\begin{proof}[Details of Equation~\eqref{eq:upperbounddim}]
	In order to obtain the dimension of the image $\Beh=\psi(\Delta^\Sscr_\Ascr)$, 
	we first look at the affine space $\aff(\Delta^\Sscr_\Ascr)$. 
	A basis of this space can be obtained by building differences of vertices of $\Delta^\Sscr_\Ascr$. 
	The vertices are the deterministic policies $\pi^f(s;a) := \delta_{f(s)}(a)$, for all $s\in\Sscr$ and $a\in\Ascr$, each of which is characterized by a function $f\colon \Sscr\to\Ascr$. 
	We fix a deterministic policy $\pi^f$, with $f(s)=a_0$ for all $s$, for some $a_0\in\Ascr$, and consider the differences 
	$e_{(s,a)} : =\pi^{f} - \pi^{f_{(s,a)}}$ for all possible pairs $(s,a)$ with $a\neq a_0$, 
	where $f_{(s,a)}$ is the function that differs from $f$ only at $s$, where it takes value $f_{(s,a)}(s)=a$. 
	This set of $e_{(s,a)}$'s is a basis of $\aff(\Delta^\Sscr_\Ascr)$. 
	There are $|\Sscr|(|\Ascr|-1)$ of these vectors, which corresponds with the dimension of $\Delta^\Sscr_\Ascr$. 
	
	Now, the image $\Beh = \psi(\Delta^\Sscr_\Ascr)$ has the dimension of $\aff(\psi(\Delta^\Sscr_\Ascr) ) = \psi(\aff\Delta^\Sscr_\Ascr)$, which is the vector space spanned by 
	$\psi( e_{(s,a)}) = \beta(w;s)(\alpha(w,a_0;dw') - \alpha(w,a ; dw'))$ for all possible pairs $(s,a)$ with $a\neq a_0$. 
	Therefore, the dimension is given by the number of linearly independent $\psi(e_{(s,a)})$'s (the rank of the matrix with rows $e_{(s,a)}$). 
\end{proof}

\begin{proof}[Details of Example~\ref{example:mineua}]
	We consider an exponential family $\Ecal$ of probability distributions on the set $\Ascr^\Sscr$ of functions $f\colon \Sscr\to\Ascr$. 
	Let this exponential family be specified by the sufficient statistic $F:=\psi\circ \eta$, 
	where $\eta\colon \Ascr^\Sscr \to \Delta^\Sscr_\Ascr$;  $f \mapsto \pi^f$, $\pi^f(s;a):=\delta_{f(s)}(a)$ for all $s\in\Sscr$ and $a\in\Ascr$, 
	and $\psi$ is the policy-behavior map, represented by the matrix $E\in\mathbb{R}^{d\times (\Sscr\times\Ascr)}$. 
	Note that, given a basis of $\aff(\psi(\Delta^\Sscr_\Ascr))$, composed of $d$ vectors in $\aff(\psi(\Delta^\Sscr_\Ascr))$, 
	we can represent each $e_{(s,a)}$ and $\psi(\pi)$ with respect to this basis by a vector of length $d\leq |\Sscr|(|\Ascr|-1)$. 
	The exponential family $\Ecal$ consists of all probability distributions of the form 
	\begin{equation*}
	p_\theta(f) = \frac{ \exp( \theta^\top  F(f) ) }{ \sum_{f'} \exp(\theta^\top F(f') )} , \quad\text{for all $f\in \Ascr^\Sscr$, for all $\theta\in\R^{d}$}. 
	\end{equation*}
	The {\em moment map} $\mu$ maps probability distributions to the corresponding expectation value of the sufficient statistic, 
	\begin{equation*}
	\mu \colon \Delta_{\Ascr^\Sscr}\to \mathbb{R}^d;\; p \mapsto \sum_{f} F(f) p(f) = \sum_f \psi(\eta(f)) p(f)  = \psi(\sum_f \eta(f) p(f)) = \psi(\pi^p), 
	\end{equation*}
	where $\pi^p:= \eta(p) =\sum_f \pi^f p(f)\in \Delta^\Sscr_\Ascr $. 
	A key property of the moment map is that it maps the closure $\overline{\Ecal}$ of $\Ecal$ bijectively to the set $\mu(\Delta_{\Ascr^\Sscr})$ of all possible expectation values. 
	We have
	\begin{equation*}
	\psi(\eta(\overline{\Ecal}) ) = \mu (\overline{\Ecal}) = \mu (\Delta_{\Ascr^\Sscr}) =\psi(\Delta^\Sscr_\Ascr). 
	\end{equation*}
	Now we only need to show that the set $\eta( \Ecal)  = \{ \pi^p = \sum_f \pi^f p(f) \colon p\in\Ecal\}$ is contained in $\Ecal^\Sscr_\Ascr$. 
	That this is true can be seen from 
	\begin{align*}
	\pi^{p_{\theta}}(s;a) 
	=& \sum_f \pi^f(s;a) p_\theta(f) 
	= \sum_{f\colon f(s)=a} p_\theta(f) 
	= p_\theta(\{f\colon f(s)=a\}) 
	= \frac{\exp(\theta^\top E(s,a) )}{\sum_{a'}\exp(\theta^\top E(s,a' ) )}\\
	=& \pi_\theta(s;a), \quad\text{for all $a\in\Ascr$ and $s\in\Sscr$}, \quad\text{for all $\theta\in\R^d$}, 
	\end{align*}
	where we used 
	\begin{align*}
	p_\theta(f) 
	=& \frac{\exp( \theta^\top E (\pi^f))}{\sum_{f'}\exp( \theta^\top E(\pi^{f'}))} 
	= \frac{\exp( \theta^\top \sum_{s} E( s,f(s) ))}
	{\sum_{f'}\exp( \theta^\top \sum_{s'} E(s',f'(s')))} 
	= \prod_{s} \frac{\exp(\theta^\top E(s,f(s)) )}{\sum_a\exp(\theta^\top E(s,a ) )}\\
	=&\prod_{s} p_\theta(\{f'  \colon f'(s) =f(s) \}). 
	\end{align*}
	In fact, since $\mu$ is a bijection between $\overline{\Ecal}$ and $\mu(\Delta_{\Ascr^\Sscr})=\psi(\Delta^\Sscr_\Ascr)$, we have that $\psi$ is a bijection between $\overline{\Ecal^\Sscr_\Ascr} = \eta(\overline{\Ecal})$ and $\psi(\Delta^\Sscr_\Ascr)$. 
\end{proof}

\begin{proof}[Proof of Lemma~\ref{proposition:euasupp}]
	Assume first that $\Scal=\Sscr$. 	
		Geometrically, the policy-behavior map $\psi$ projects the policy polytope linearly into a polytope of dimension $d$. 
		A $d$-dimensional projection of a polytope is equal to the projection of its $d$-dimensional faces. 
		This implies that the result of applying any given policy from $\Delta_\Ascr^\Sscr$ can be achieved equally well by applying a policy from a $d$-dimensional face of $\Delta_\Ascr^\Sscr$. 
See Figure~\ref{fig:facesprojection} for an illustration of what we mean. 

		Now, the $d$-dimensional faces of the policy polytope $\Delta^\Sscr_\Ascr$ consist of those policies with at most $|\Sscr| + d$ non-zero entries. 
		The arguments for this are as follows. 
		The policy polytope is a product of simplices $\Delta_\Ascr^\Sscr=\times_{s\in\Sscr} \, \Delta_\Ascr$, 
		where the $s$-th factor corresponds to the set of all possible probability distributions $\pi(s;\cdot)$. 
		The faces of $\Delta_\Ascr^\Sscr$ are products of faces of its factors and have the form $\times_{s\in\Sscr} \Delta_{\Ascr_s}$,  
		where $\Ascr_s\subseteq\Ascr$ for all $s\in\Sscr$. 
		Each face of $\Delta_\Ascr^\Sscr$ corresponds to a choice of positions $\Ascr_s\subseteq\Ascr$ of the non-zero entries of $\pi(s;\cdot)$ for all $s\in\Sscr$. 
		The $d$-dimensional faces are those for which $\sum_{s\in\Sscr}(|\Ascr_s|-1) = d$, meaning that they consist of policies which have at most $\sum_{s\in\Sscr}|\Ascr_s| = |\Sscr| + d$ non-zero entries. 
		
	Consider now  $\Scal\subseteq\Sscr$. 
	The projection of the policy polytope by the $\mathcal{S}$ embodiment matrix can be regarded as 
	a composition which first projects $\Delta_\Ascr^\Sscr$ to $\Delta_{\Ascr}^{\mathcal{S}}$ and then projects $\Delta_{\Ascr}^{\mathcal{S}}$ by $\psi^\Scal$. 
	Now we can use the same arguments as above, with the difference that now only need to represent the $d^{\mathcal{S}}$-dimensional faces of the polytope $\Delta_\Ascr^{\mathcal{S}}$. 
\end{proof}

\begin{figure}
	\centering
	\includegraphics{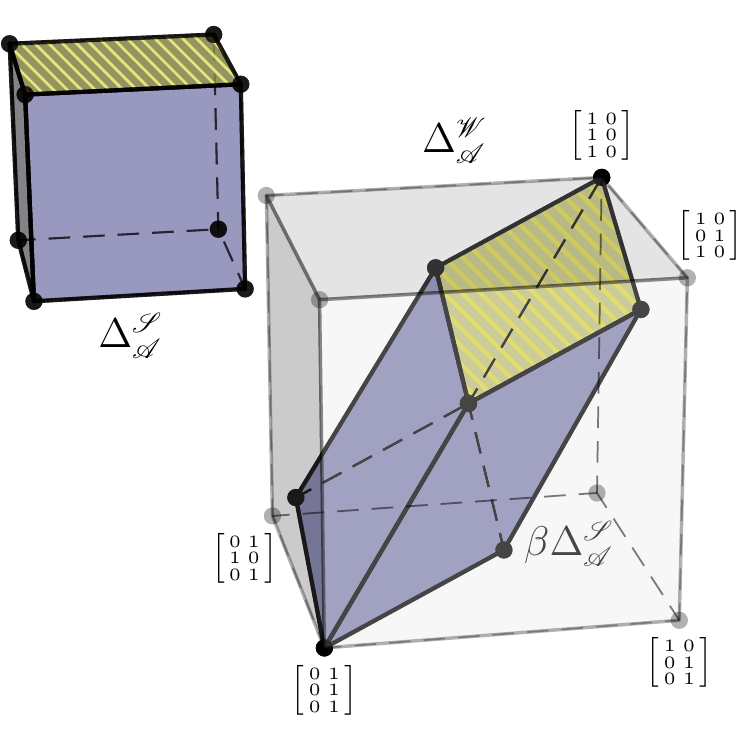}
	\caption{
		Illustration of the projection of the set of policies $\Delta^\Sscr_\Ascr$ into a lower dimensional polytope. 
        In this example $|\Sscr|=3$ and $|\Ascr|=2$, such that the policy polytope $\Delta^\Sscr_\Ascr$ is a cube. 
        The set of all kernels $\Delta_\Ascr^\Wscr$, $|\Wscr|=3$, is also a cube. 
        The set of all realizable kernels within that set, $p(w; a) = \sum_{s} \beta(w;s)\pi(s;a)$, for all choices of the policy $\pi$, is the projection of $\Delta_{\mathscr{A}}^{\mathscr{S}}$ by $\beta$. 
        In this example $\rank(\beta)=2$ and the projection of the policy polytope is two-dimensional polygon (the blue hexagon). 
        This projection by $\beta$ represents one part of the projection by the policy-behavior map. 
		The two-dimensional faces of the policy polytope (one of them highlighted in dashed yellow) have the same image as the entire policy polytope. 
	}
	\label{fig:facesprojection}
\end{figure}

\section{Estimation of the Embodied Behavior Dimension based on the Internal World Model}
\label{sec:exp embodiment dimension}
In many situations, the embodied behavior dimension is not available from a perspective
that is intrinsic to the agent, as the agent does not have direct access to the sensor kernel
$\beta$ nor to the world kernel $\alpha$. From that perspective, only an internal
version of the world model is accessible, which we refer to as {\em internal
	world model}. It is defined as a kernel $\gamma \in \Delta^{\mathscr{S} \times
	\mathscr{A}}_{\mathscr{S}}$, assigning to each sensor state $s$ with positive
probability and each actuator state $a$ the next sensor state $s'$, that is, 
\begin{eqnarray}
\gamma(s,a; s') & = & \int_{\Wscr} \left\{ \int_{\Wscr}  \beta(w'; s')\,  \alpha(w,a ; dw') \right\} {\Bbb P}(s; dw) \, , 
\label{eq:supp_delta}
\end{eqnarray}
where 
\[
{\Bbb P}(s; dw) \; := \;   \frac{ \beta(w; s) }{\int_{\Wscr}  \beta(w'' ; s ) \, {\Bbb P}(dw'')} \, {\Bbb P}(dw) \, .
\]
Note that the internal world model is not completely determined by $\beta$ and $\alpha$. It also depends on the distribution ${\Bbb P}(dw)$ 
of the world states $w$. If we choose this distribution to be a fixed reference distribution of world states, then the world model will be determined by $\alpha$ and $\beta$ only. However, in order to describe the actual distribution of world states, we have to take into account the contribution  
of the agent's policy $\pi$. This implies that, if the policy is subject to changes in terms of a learning process, then, in general, 
the world model will also be time dependent.  

On the other hand, $\gamma$ is the only information about world dynamics
that is intrinsically available to the agent. The extent 
to which $\gamma$ is not a good replacement for $\alpha$ depends on how much the
agent can ``see'' from the world with its sensors. If the agent has direct
access to the world state, that is $W^t = S^t$, then $\alpha$ and $\gamma$ coincide.
However, this is not very realistic. Generically, only partial observation of
the world is possible. Now the question arises whether it is possible to
determine the embodied behavior dimension $d$ in terms of $\gamma$
even in cases where the agent has only partial access to the world state. This
is indeed possible under specific conditions which are satisfied in our experimental
setup. We first present these conditions in general terms,
before we then relate them to our experiment at the end of this section. 
Let the world state $w$ consist of two parts $s$ and $r$, 
where $s$ is directly accessible to the agent and $r$ is the remaining part of the world, which is hidden to the agent. The situation is illustrated in~Figure~\ref{fig:specialstructure}.
\begin{figure}[h]
	\centering
		\includegraphics[scale=1.14]{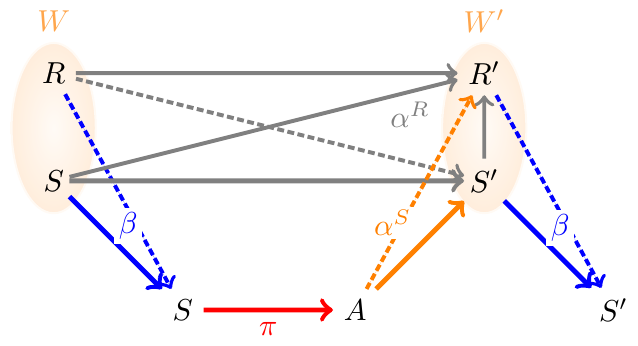}
	\caption{Special causal structure of the sensorimotor loop. 
		The dashed arrows are the ones that we omit within our assumptions. 
	}\label{fig:specialstructure}
\end{figure}

In this interpretation of the world state, the sensor
kernel $\beta$ is simply the identity map~$s \mapsto s$. Furthermore, this
interpretation sets structural constraints on the world transition kernel
$\alpha$, which assigns a probability distribution of the next world state $w' = (s',r')$
given the current world state $w = (s,r)$ and an actuator value $a$. As $r$ is
assumed to be hidden to the agent, $s'$ should not depend on $r$. This leads to
the following natural factorization of $\alpha$:
\begin{equation}
\alpha(s,r, a ; s', dr') \, = \, \alpha^{S}(s,a ; s') \cdot \alpha^R(r,s,s'; d r').
\end{equation}                 
With this assumption, we obtain as internal world model
\begin{equation}
\gamma( s,a ; s') \, = \, \alpha^S(s , a; s'), \qquad \mbox{whenever ${\Bbb P}(s) > 0$},    
\end{equation}
and the following transition probabilities from a world state $w = (s,r)$ with positive probability to a
world states $w' = (s',r')$:
\begin{eqnarray}
{\Bbb P}^\pi(s,r ; s', dr') 
& = & \sum_{{s}'', {a}} \beta(s,r ; s'') \, \pi( {s}'' ; a) \, \alpha(s,r, a ; s', dr') \nonumber\\
& = & \sum_{{a}}  \pi(s ; {a}) \, \alpha^{S}(s,{a} ; s' ) \, \alpha^R(r,s,s' ; dr') \\
& = & \alpha^R( r,s,s' ; dr')  \underbrace{\sum_{{a}}  \pi(s ; a) \, \alpha^S ( s,{a} ; s' )}_{=: \; {\Bbb Q}^\pi(s ; s')}.\nonumber
\end{eqnarray}
This shows that  ${\Bbb P}^\pi =  {\Bbb P}^{\pi^\ast}$ if and only if 
${\Bbb Q}^{\pi} = {\Bbb Q}^{\pi^\ast}$, and therefore the embodied behavior dimension is given by the dimension
of the image of $\pi \mapsto {\Bbb Q}^\pi$. This is given by the affine rank of the
kernel $\alpha^S$, which coincides with the affine rank of $\gamma$:  
\begin{equation}
d = 
      \sum_s \rank( (\gamma(s,a_0 ; s') - \gamma(s, a ; s'))_{s'\in\Sscr, a\in\Ascr} ),  \label{eq:affine rank}
\end{equation}
where $a_0$ is any fixed value in $\Ascr$, for all~$s$. 

This applies to our hexapod experiment discussed in
Section~\ref{section:experiments} for the following reason.
In the special case of the tripod gait of a hexapod on an even and otherwise
featureless plane, the next joint angles $S^{t+1}$ are only determined by the
current joint angles $S^t$ and the current action $A^t$. The rest of the world,
here denoted by $R$, contains information such the contact points of the legs
with the ground.
This information is carried from one time step to the next, as it determines how the
hexapod walks along the plane. Nevertheless, the contact points of the legs do not influence
the joint angles. Hence, in our experiment, $S'$ is conditionally independent of
$R$ given $S$ and $A$. Furthermore, $R'$ is conditionally independent of
$A$ given $R$, $S$ and $S'$ as the contact points of the
legs with the ground are only determined by the relative joint angles, and not by
the current action. Therefore, we can estimate the embodied behavior dimension by the
rank of the internal world model.

\section{Generalizations}
\label{sec:generalizations}

\subsection{SML with Internal State}

In the main part of the paper we considered reactive SMLs. 
In a more general setting, the agent may be equipped 
with some sort of memory or internal representation of the world. 
In this case, besides from the world state, the sensor state, and the actuator state, the SML also includes an internal state variable. 
As in the reactive SML, the dynamics of these variables are governed by Markov transition kernels, but the causality structure is slightly different. 
Let $\mathscr{W}$, $\mathscr{S}$, $\mathscr{C}$, and $\mathscr{A}$ denote the sets of possible states of the world, the sensors, the internal state, and the actuators. 
Then the Markov kernels are 
\begin{equation}
\def\arraystretch{1.5}
\begin{array}{l c r}
\beta      \colon \mathscr{W}  &\to& \Delta_{\Sscr}, \\ 
\varphi   \colon \mathscr{C}  \times \mathscr{S} &\to& \Delta_{\Cscr}, \\
\pi          \colon \mathscr{C}  &\to& \Delta_{\Ascr}, \\
\alpha  \colon \mathscr{W} \times \mathscr{A} &\to& \Delta_{\Wscr}. 
\end{array}\label{eq:kernels}
\end{equation}
See Figure~\ref{figure:SMLin} for an illustration of this causality structure. 

\begin{figure}[h]
	\centering 
	\includegraphics[scale=1.14]{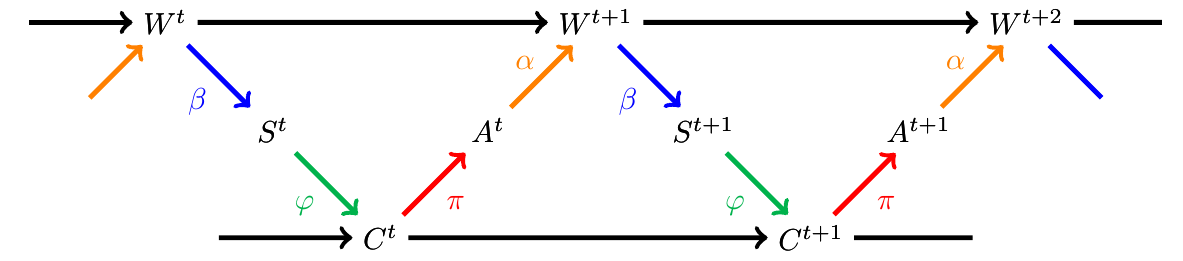}
\caption{Causal structure of a SML with internal state. Here $W^t, S^t, C^t, A^t$ are the states of the world, sensors, internal variable, and actuators at the discrete time $t$. }
	\label{figure:SMLin}
\end{figure}

As in the reactive case discussed in the main text, here we also want to consider the (extrinsic) behavior of the agent, which is described in terms of the stochastic process of world states. 
In this case, however, we condition these processes not only on an initial world state but also on an initial internal state. 
The difficulty arising here is that, in general, in the presence of an internal state the stochastic process over the world states is not Markovian. 
The world state transition at each time step does not only depend on the previous world state but it also depends on a longer history, encoded in the internal state. 

For example, when navigating a territory, a robot endowed with an internal state could operate in the following way. 
If at a given time step the robot detects an obstacle ahead, $s =\text{``obstacle''}$, 
then, in conjunction with a current internal state $c =\text{``safe''}$, the new internal state could become $c'=\text{``attentive''}$, in which case the policy would choose $a =\text{``maintain direction''}$. 
However, if the current internal state was $c=\text{``attentive''}$, 
the new internal state could become $c'=\text{``alert''}$, 
in which case the policy would choose between the actions $a'=\text{``turn left''}$ and $a'=\text{``turn right''}$ with probability $\tfrac12$. 
This example shows that the internal state may contain information about the history of world and sensor states, 
which is not available from the current world and sensor states alone. 

Nonetheless, for any fixed choice of the kernels $\beta, \varphi, \pi, \alpha$ and a starting value $ (w^0, c^0)$ at time $t=0$, the SML defines a (discrete-time homogeneous) Markov chain with state space $\Wscr \times \Sscr \times \Cscr \times \Ascr$. 
The transition probabilities of this chain are given by  
\begin{equation*}
\mathbb{P}^\pi(w^0, s^0, c^0, a^0 ;  dw^1, s^1, c^1, a^1) = 
\alpha(w^0,a^0; dw^1) \beta(w^1; ds^1)\varphi(c^0,s^1 ; dc^1) \pi(c^1; da^1) .
\end{equation*}
Furthermore, the process with state space $\Wscr\times\Cscr$ is also Markovian.  
The transition probabilities of this chain 
are given by 
\begin{equation*}
\psi(w,c ; dw', dc') = 
\int_{\Sscr'} \int_\Ascr  \pi(c; da)\alpha(w,a; dw')\beta(w'; ds')\varphi(c,s'; dc') . 
\end{equation*}

The process on $\Wscr$ (extrinsic behavior) is the marginal of the process on $\Wscr\times\Cscr$ (extrinsic-intrinsic behavior).  
We can study some properties of the extrinsic behavior in terms of the properties of the extrinsic-intrinsic behavior.  
The latter is easier to analyze, since it is Markovian. In particular, we can study it in the same way we studied the extrinsic behavior in the reactive SML. 

More explicitly we have the following.  
Writing $\xi(w,c; dw') = \int_{\Ascr} \pi(c; da) \alpha(w,a; dw')$ and $\phi(w',c; dc') = \int_{\Sscr'}\beta(w';ds') \varphi(c,s';dc')$, the transition probabilities for the process on $\Wscr\times\Cscr$ are given by 
\begin{equation}
\psi(w,c ; dw', dc') = 
\xi(w,c; dw') \phi(w',c; dc'). 
\end{equation}
For each $(w,c)$, the probability distribution $\xi(w,c;\cdot)\in\Delta_\Wscr$ is the projection of $\pi(c;\cdot)\in\Delta_\Ascr$ by the linear map defined by $\alpha(w,\cdot;\cdot)$. 
If the intersection of the null-spaces of $\alpha(w,\cdot;\cdot)$ for all $w$ has a positive dimension, 
then there is a positive dimensional set of policies $\pi$ that are mapped to the same $\xi$ and hence to the same behavior. 
In order to obtain that behavior, it is sufficient to represent one of the policies that map to $\xi$, in contrast to the potentially much larger set of all policies that map to the same $\xi$. 
A similar observation applies to $\varphi$. 
This shows that already when considering the process on $\Wscr\times\Cscr$ (the combined extrinsic-intrinsic behavior of the agent), many policies may be identified. Embodiment constraints restrict the possible behaviors. 
When considering only the process over $\Wscr$ (the extrinsic behavior of the agent), many more policies may be identified with the same behavior. The detailed study of projections from combined behaviors to extrinsic behaviors is left for future work. 

\subsection{Continuous Sensor and Actuator State Spaces}
We have considered systems where $\Sscr$ and $\Ascr$ are finite sets. 
In some case it can be more natural to consider continuous sensor and actuator spaces. 
The continuous case brings some subtleties with it. 
In particular, the set of policies with continuous state spaces is infinite dimensional. 
In this case one has to depart from linear algebra and use functional analysis. 
Furthermore, in the setting of continuous sensor and actuator spaces usually it is not possible to achieve universal approximation by one fixed model. 
Rather, one says that a class of models has the universal approximation property, meaning that 
for each given error tolerance, there is a model in that class, that can approximate to within that error tolerance. 
Nonetheless, one can measure the approximation performance in terms of the (finite) number of parameters or hidden variables that a model needs in order to satisfy a given error tolerance. 
Continuous policy models can be defined in terms of stochastic feedforward neural networks with continuous variables or also in terms of CRBMs with Gaussian output units. 
Here, the complexity of a model can be measured in terms of the number of hidden variables.

\subsubsection*{Acknowledgment}
We would like to acknowledge support for this project from the DFG Priority Program Autonomous Learning (DFG-SPP 1527). 
G.~M. and K.~G.-Z. would like to thank the Santa Fe Institute for hosting them during the work on this article. 

\bibliographystyle{abbrvnat}
\bibliography{referenzen}{}

\begin{thebibliography}{35}
\providecommand{\natexlab}[1]{#1}
\providecommand{\url}[1]{\texttt{#1}}
\expandafter\ifx\csname urlstyle\endcsname\relax
  \providecommand{\doi}[1]{doi: #1}\else
  \providecommand{\doi}{doi: \begingroup \urlstyle{rm}\Url}\fi

\bibitem[Ay and
  Zahedi(2014)]{Ay2014On-the-causal-structure-of-the-sensorimotor}
N.~Ay and K.~Zahedi.
\newblock On the causal structure of the sensorimotor loop.
\newblock In M.~Prokopenko, editor, \emph{Guided Self-Organization: Inception}.
  Springer, 2014.

\bibitem[Bauer(1996)]{bauer1996probability}
H.~Bauer.
\newblock \emph{Probability Theory}.
\newblock De Gruyter studies in mathematics. Bod Third Party Titles, 1996.
\newblock ISBN 9783110139358.
\newblock URL \url{http://books.google.com/books?id=w76IHsPHybcC}.

\bibitem[Bengio(2009)]{BengioLearning}
Y.~Bengio.
\newblock Learning deep architectures for ai.
\newblock \emph{Foundations and Trends® in Machine Learning}, 2\penalty0
  (1):\penalty0 1--127, 2009.
\newblock ISSN 1935-8237.
\newblock \doi{10.1561/2200000006}.
\newblock URL \url{http://dx.doi.org/10.1561/2200000006}.

\bibitem[Bialek and Tishby(1999)]{Bialek1999aPredictive}
W.~Bialek and N.~Tishby.
\newblock Predictive information.
\newblock SEE Bialek, Nemenman, Tishby, 2001, 1999.

\bibitem[Braitenberg(1984)]{Braitenberg1984aVehicles}
V.~Braitenberg.
\newblock \emph{Vehicles}.
\newblock MIT Press, Cambridge MA, 1984.

\bibitem[Brooks(1989)]{Brooks1989aA-Robot}
R.~A. Brooks.
\newblock A robot that walks; emergent behaviors from a carefully evolved
  network.
\newblock \emph{Neural Comput.}, 1\penalty0 (2):\penalty0 253--262, June 1989.
\newblock ISSN 0899-7667.
\newblock \doi{10.1162/neco.1989.1.2.253}.
\newblock URL \url{http://dx.doi.org/10.1162/neco.1989.1.2.253}.

\bibitem[Brooks(1991{\natexlab{a}})]{Brooks1991aIntelligence}
R.~A. Brooks.
\newblock Intelligence without reason.
\newblock In J.~Myopoulos and R.~Reiter, editors, \emph{Proceedings of the 12th
  International Joint Conference on Artificial Intelligence ({IJCAI}-91)},
  pages 569--595, Sydney, Australia, 1991{\natexlab{a}}. Morgan Kaufmann
  publishers Inc.: San Mateo, CA, USA.

\bibitem[Brooks(1991{\natexlab{b}})]{Brooks1991bIntelligence}
R.~A. Brooks.
\newblock Intelligence without representation.
\newblock \emph{Artificial Intelligence}, 47\penalty0 (1-3):\penalty0 139--159,
  1991{\natexlab{b}}.

\bibitem[Clark and Sokoloff(1999)]{Clark1999aCirculation}
D.~D. Clark and L.~Sokoloff.
\newblock Circulation and energy metabolism of the brain.
\newblock In G.~J. Siegel, B.~W. Agranoff, R.~W. Albers, S.~K. Fisher, and
  M.~D. Uhler, editors, \emph{Basic Neurochemistry: Molecular, Cellular and
  Medical Aspects}, chapter~31. Lippincott-Raven, Philadelphia, 6th edition,
  1999.

\bibitem[Coumans(2012)]{Coumans2012aBullet}
E.~Coumans.
\newblock Bullet physic sdk manual.
\newblock \url{www.bulletphysics.org}, 2012.

\bibitem[Freund and Haussler(1994)]{freund1994unsupervised}
Y.~Freund and D.~Haussler.
\newblock \emph{Unsupervised Learning of Distributions of Binary Vectors Using
  Two Layer Networks}.
\newblock Technical report. Computer Research Laboratory, University of
  California, Santa Cruz, 1994.

\bibitem[Hinton(2012)]{Hinton:practical}
G.~Hinton.
\newblock A practical guide to training restricted {B}oltzmann machines.
\newblock In G.~Montavon, G.~Orr, and K.-R. Müller, editors, \emph{Neural
  Networks: Tricks of the Trade}, volume 7700 of \emph{Lecture Notes in
  Computer Science}, pages 599--619. Springer Berlin Heidelberg, 2012.
\newblock ISBN 978-3-642-35288-1.
\newblock \doi{10.1007/978-3-642-35289-8_32}.
\newblock URL \url{http://dx.doi.org/10.1007/978-3-642-35289-8_32}.

\bibitem[Hinton(2002)]{Hinton:2002:TPE:639729.639730}
G.~E. Hinton.
\newblock Training products of experts by minimizing contrastive divergence.
\newblock \emph{Neural Computation}, 14\penalty0 (8):\penalty0 1771--1800,
  2002.

\bibitem[Klyubin et~al.(2004)Klyubin, Polani, and Nehaniv]{klyubin2004tracking}
A.~S. Klyubin, D.~Polani, and C.~L. Nehaniv.
\newblock Tracking information flow through the environment: Simple cases of
  stigmerg.
\newblock In J.~Pollack, editor, \emph{Artificial Life IX: Proceedings of the
  Ninth International Conference on the Simulation and Synthesis of Living
  Systems}. MIT Press, 2004.

\bibitem[Larochelle and Bengio(2008)]{LarochelleB08}
H.~Larochelle and Y.~Bengio.
\newblock Classification using discriminative restricted {B}oltzmann machines.
\newblock In W.~W. Cohen, A.~McCallum, and S.~T. Roweis, editors,
  \emph{Proceedings of the 25th International Conference on Machine Learning
  (ICML 2008)}, volume 307, pages 536--543, 2008.

\bibitem[Le~Roux and Bengio(2008)]{LeRoux:2008:RPR:1374176.1374187}
N.~Le~Roux and Y.~Bengio.
\newblock Representational power of restricted {B}oltzmann machines and deep
  belief networks.
\newblock \emph{Neural Computation}, 20\penalty0 (6):\penalty0 1631--1649, June
  2008.

\bibitem[Long and Servedio(2010)]{LongServedio10}
P.~M. Long and R.~A. Servedio.
\newblock Restricted {B}oltzmann machines are hard to approximately evaluate or
  simulate.
\newblock In J.~F{\"u}rnkranz and T.~Joachims, editors, \emph{Proceedings of
  the 27th International Conference on Machine Learning (ICML 2010)}, pages
  703--710. Omnipress, 2010.

\bibitem[Lungarella and Sporns(2005)]{Lungarella2005aInformation}
M.~Lungarella and O.~Sporns.
\newblock Information self-structuring: Key principle for learning and
  development.
\newblock In IEEE, editor, \emph{Proceedings. The 4th International Conference
  on Development and Learning, 2005.}, pages 25--30, San Diego, CA, 2005. IEEE
  Press.

\bibitem[Martens et~al.(2013)Martens, Chattopadhya, Pitassi, and
  Zemel]{NIPS2013_5020}
J.~Martens, A.~Chattopadhya, T.~Pitassi, and R.~Zemel.
\newblock On the expressive power of restricted {B}oltzmann machines.
\newblock In C.~Burges, L.~Bottou, M.~Welling, Z.~Ghahramani, and
  K.~Weinberger, editors, \emph{NIPS 26}, pages 2877--2885, 2013.

\bibitem[McGeer(1990)]{McGeer1990aPassive}
T.~McGeer.
\newblock Passive dynamic walking.
\newblock \emph{International Journal of Robotic Research}, 9\penalty0
  (2):\penalty0 62--82, 1990.

\bibitem[Mont{\'u}far and Ay(2011)]{Montufar2011}
G.~Mont{\'u}far and N.~Ay.
\newblock Refinements of universal approximation results for deep belief
  networks and restricted {B}oltzmann machines.
\newblock \emph{Neural Computation}, 23\penalty0 (5):\penalty0 1306--1319,
  2011.

\bibitem[Mont\'ufar et~al.(2011)Mont\'ufar, Rauh, and Ay]{NIPS2011_0307}
G.~Mont\'ufar, J.~Rauh, and N.~Ay.
\newblock Expressive power and approximation errors of restricted {B}oltzmann
  machines.
\newblock In J.~Shawe-Taylor, R.~Zemel, P.~Bartlett, F.~Pereira, and
  K.~Weinberger, editors, \emph{NIPS 24}, pages 415--423, 2011.

\bibitem[Mont{\'u}far et~al.(2014)Mont{\'u}far, Ay, and
  Ghazi-Zahedi]{montufar2014CRBMs}
G.~Mont{\'u}far, N.~Ay, and K.~Ghazi-Zahedi.
\newblock Expressive power of conditional restricted {B}oltzmann machines.
\newblock \emph{arXiv preprint arXiv:1402.3346}, 2014.

\bibitem[Pfeifer and Bongard(2006)]{Pfeifer2006aHow-the-Body}
R.~Pfeifer and J.~C. Bongard.
\newblock \emph{How the Body Shapes the Way We Think: A New View of
  Intelligence}.
\newblock The MIT Press (Bradford Books), Cambridge, MA, 2006.

\bibitem[Pfeifer et~al.(2007)Pfeifer, Lungarella, and
  Iida]{Pfeifer2007aSelf-Organization}
R.~Pfeifer, M.~Lungarella, and F.~Iida.
\newblock Self-organization, embodiment, and biologically inspired robotics.
\newblock \emph{Science}, 318\penalty0 (5853):\penalty0 1088--1093, 2007.

\bibitem[Polani et~al.(2006)Polani, Nehaniv, Martinetz, and
  Kim]{Polani2006aRelevant}
D.~Polani, C.~Nehaniv, T.~Martinetz, and J.~T. Kim.
\newblock Relevant {I}nformation in {O}ptimized {P}ersistence vs. {P}rogeny
  {S}trategies.
\newblock In L.~M. Rocha, M.~Bedau, D.~Floreano, R.~Goldstone, A.~Vespignani,
  and L.~Yaeger, editors, \emph{Proc. Artificial Life X}, pages 337--343,
  Cambridge, MA, 2006. MIT Press.

\bibitem[Salakhutdinov et~al.(2007)Salakhutdinov, Mnih, and
  Hinton]{Salakhutdinov:2007:RBM}
R.~Salakhutdinov, A.~Mnih, and G.~E. Hinton.
\newblock Restricted {B}oltzmann machines for collaborative filtering.
\newblock In \emph{Proceedings of the 24th International Conference on Machine
  Learning (ICML 2007)}, pages 791--798, 2007.

\bibitem[Sallans and Hinton(2004)]{Sallans:2004:RLF:1005332.1016794}
B.~Sallans and G.~E. Hinton.
\newblock Reinforcement learning with factored states and actions.
\newblock \emph{J. Mach. Learn. Res.}, 5:\penalty0 1063--1088, 2004.

\bibitem[Schmidhuber(2009)]{Schmidhuber2009aDriven}
J.~Schmidhuber.
\newblock Driven by compression progress: A simple principle explains essential
  aspects of subjective beauty, novelty, surprise, interestingness, attention,
  curiosity, creativity, art, science, music, jokes.
\newblock \emph{Anticipatory Behavior in Adaptive Learning Systems}, pages
  48--76, 2009.

\bibitem[Smolensky(1986)]{Smolensky1986}
P.~Smolensky.
\newblock Information processing in dynamical systems: foundations of harmony
  theory.
\newblock In \emph{Symposium on Parallel and Distributed Processing}, 1986.

\bibitem[Sokoloff et~al.(1955)Sokoloff, Mangold, Wechsler, Kennedy, and
  Kety]{SOKOLOFF1955aEFFECT}
L.~Sokoloff, R.~Mangold, R.~Wechsler, C.~Kennedy, and S.~Kety.
\newblock Effect of mental arithmetic on cerebral circulation and metabolism.
\newblock \emph{J. Clin. Invest.}, 34\penalty0 (7):\penalty0 1101--1108, 1955.

\bibitem[Sol et~al.(2010)Sol, Garcia, Iwaniuk, Davis, Meade, Boyle, and
  Sz{\'e}kely]{Sol2010aEvolutionary}
D.~Sol, N.~Garcia, A.~Iwaniuk, K.~Davis, A.~Meade, W.~A. Boyle, and
  T.~Sz{\'e}kely.
\newblock Evolutionary divergence in brain size between migratory and resident
  birds.
\newblock \emph{PLoS ONE}, 5\penalty0 (3):\penalty0 e9617, 03 2010.

\bibitem[Sutskever and Hinton(2007)]{sutskever_hinton_07}
I.~Sutskever and G.~E. Hinton.
\newblock Learning multilevel distributed representations for high-dimensional
  sequences.
\newblock \emph{Proceeding of the 11th International Conference on Artificial
  Intelligence and Statistics}, 2007.

\bibitem[Taylor et~al.(2007)Taylor, Hinton, and Roweis]{Taylor06modelinghuman}
G.~W. Taylor, G.~E. Hinton, and S.~Roweis.
\newblock Modeling human motion using binary latent variables.
\newblock In \emph{NIPS 19}, pages 1345--1352. MIT Press, 2007.

\bibitem[Zahedi et~al.(2008)Zahedi, von Twickel, and
  Pasemann]{Zahedi2008aYARS:}
K.~Zahedi, A.~von Twickel, and F.~Pasemann.
\newblock Yars: A physical 3d simulator for evolving controllers for real
  robots.
\newblock In S.~Carpin, I.~Noda, E.~Pagello, M.~Reggiani, and O.~von Stryk,
  editors, \emph{SIMPAR 2008}, LNAI 5325, pages 71--82. Springer, 2008.

\end{thebibliography}

\end{document}